\documentclass[11pt]{article}
\usepackage[letterpaper, margin=1in]{geometry}
\usepackage[T1]{fontenc}
\usepackage{amsfonts,amsmath,amsthm,amssymb,mathtools}  %

\mathtoolsset{centercolon}
\usepackage{xfrac,nicefrac}
\usepackage{xcolor}
\usepackage{mathdots}
\usepackage{mleftright}  %
\let\left\mleft
\let\right\mright

\usepackage{xspace}
\xspaceaddexceptions{]\}}  %
\usepackage{regexpatch}

\usepackage{bm,bbm,dsfont}  %
\usepackage{caption}
\usepackage[normalem]{ulem}
\usepackage{enumitem}

\usepackage{graphicx}
\usepackage{float}
\usepackage{subcaption}  %
\usepackage{tcolorbox}
\usepackage{tikz}
\usetikzlibrary{decorations.pathreplacing}
\usetikzlibrary{calc}
\usetikzlibrary{positioning}
\usetikzlibrary{arrows.meta}

\usepackage[linesnumbered,boxed,ruled,vlined]{algorithm2e}
\usepackage{algpseudocode}

\usepackage{thmtools,thm-restate}
\theoremstyle{plain}

\newtheorem{theorem}{Theorem}[section]  %
\newtheorem{lemma}[theorem]{Lemma}
\newtheorem{fact}[theorem]{Fact}

\newtheorem{cor}[theorem]{Corollary}
\newtheorem{claim}[theorem]{Claim}
\theoremstyle{definition}  %
\newtheorem{definition}[theorem]{Definition}
\newtheorem{remark}[theorem]{Remark}

\newenvironment{proofof}[1]{\begin{proof}[Proof of #1]}{\end{proof}}

\usepackage[colorlinks,citecolor=blue,linkcolor=blue,urlcolor=red]{hyperref}
\usepackage[capitalise]{cleveref}
\crefname{algocf}{Algorithm}{Algorithms}
\Crefname{algocf}{Algorithm}{Algorithms}
\crefname{prob}{Problem}{Problems}
\crefname{claim}{Claim}{Claims}
\crefname{cor}{Corollary}{Corollaries}
\crefname{fact}{Fact}{Facts}

\DeclarePairedDelimiter{\norm}{\lVert}{\rVert}
\DeclarePairedDelimiter{\bk}{(}{)}
\DeclarePairedDelimiter{\Bk}{[}{]}
\DeclarePairedDelimiter{\BK}{\{}{\}}

\DeclarePairedDelimiter{\angbk}{\langle}{\rangle}
\DeclarePairedDelimiter{\abs}{\lvert}{\rvert}

\DeclareMathOperator*{\E}{\mathbb{E}}

\let\Pr\PrAux
\DeclareMathOperator{\poly}{poly}

\newcommand{\F}{\mathbb{F}}

\renewcommand{\tilde}{\widetilde}
\newcommand{\eqdef}{\eqqcolon}
\newcommand{\defeq}{\coloneqq}
\newcommand{\eps}{\varepsilon}

\newcommand{\N}{\mathbb{N}}
\newcommand{\R}{\mathbb{R}}
\newcommand{\Z}{\mathbb{Z}}
\renewcommand{\l}{\ell}

\renewcommand{\epsilon}{\eps}

\newcommand{\numberthis}{\addtocounter{equation}{1}\tag{\theequation}}

\usepackage{regexpatch}
\makeatletter
\xpatchcmd\thmt@restatable{%
\csname #2\@xa\endcsname\ifx\@nx#1\@nx\else[{#1}]\fi
}{%
\ifthmt@thisistheone
\csname #2\@xa\endcsname\ifx\@nx#1\@nx\else[{#1}]\fi
\else
\csname #2\@xa\endcsname[{Restated}]
\fi}{}{}
\makeatother

\newcommand{\defn}[1]{\textbf{\emph{#1}}}

\newcommand{\nnz}{\textup{nnz}}
\newcommand{\rank}[1][\F_p]{\textup{rank}^{#1}}

\renewcommand{\R}{\mathcal{R}}

\newcommand{\bool}{\textup{bool}}

\renewcommand{\vec}[1]{\bm{\mathrm{#1}}}

\newcommand{\circledM}{\tikz[baseline=(char.base)]{\node[shape=circle,draw,inner sep=0.4pt,minimum size=0.4em] (char) {\fontsize{3}{4}\selectfont M};}}
\newcommand{\maj}[2][n]{#2^{\circledM #1}}
\newcommand{\kro}[2][n]{#2^{\otimes #1}}
\newcommand{\Maj}{\textup{Maj}}

\newcommand{\indicator}[1]{\mathbbm{1}{\Bk*{#1}}}
\newcommand{\boolRigidity}[2]{\R_{#1}^{\F_p^{\bool}}\bk*{#2}}

\newcommand{\pre}{\text{pre}}

\newcommand{\C}{\mathbb{C}}

\title{Low Rank Matrix Rigidity: \\ Tight Lower Bounds and Hardness Amplification}
\author{Josh Alman\thanks{Columbia 
University. \texttt{josh@cs.columbia.edu}. Work supported in part by NSF Grant CCF-2238221.} 
    \and 
    Jingxun Liang\thanks{Carnegie Mellon University. \texttt{jingxunl@andrew.cmu.edu}. Work performed while the author was visiting Columbia University and supported in part by NSF Grant CCF-2238221.}}
\date{}

\begin{document}

\maketitle

\begin{abstract}
For an $N \times N$ matrix $A$, its rank-$r$ rigidity, denoted $\R_A(r)$, is the minimum number of entries of $A$ that one must change to make its rank become at most $r$. Determining the rigidity of interesting explicit families of matrices remains a major open problem, and is central to understanding the complexities of these matrices in many different models of computation and communication. We focus in this paper on the Walsh-Hadamard transform and on the `distance matrix', whose rows and columns correspond to binary vectors, and whose entries calculate whether the row and column are close in Hamming distance. Our results also generalize to other Kronecker powers and `Majority powers' of fixed matrices. We prove two new results about such matrices.

First, we prove new rigidity lower bounds in the low-rank regime where $r < \log N$. For instance, we prove that over any finite field, there are constants $c_1, c_2 > 0$ such that the $N \times N$ Walsh-Hadamard matrix $H_n$ satisfies $$\R_{H_n}(c_1 \log N) \geq N^2 \left( \frac12 - N^{-c_2} \right),$$ and a similar lower bound for the other aforementioned matrices. This is tight, and is the new best rigidity lower bound for an explicit matrix family at this rank; the previous best was $\R(c_1 \log N) \geq c_3 N^2$ for a small constant $c_3>0$. It also rules out an improvement to the current best constant-depth linear circuits for $H_n$ and other Kronecker powers by using asymptotic improvements to low-rank rigidity upper bounds; the current best uses the rank-1 rigidity of constant-sized matrices.

Second, we give new hardness amplification results, showing that rigidity lower bounds for these matrices for slightly higher rank would imply breakthrough rigidity lower bounds for much higher rank. For instance, if one could prove $$\R_{H_n}(\log^{1 + \varepsilon} N) \geq N^2 \left( \frac12 - N^{-1/2^{(\log \log N)^{o(1)}}} \right)$$ over any finite field for some $\varepsilon>0$, this would imply that $H_n$ is \emph{Razborov rigid}, giving a breakthrough lower bound in communication complexity. This may help explain why our best rigidity lower bounds for explicit matrices are stuck: there seems to be a big gap between what is known and what is required in the rank regime for Razborov rigidity, but our new required lower bound for rank $\log^{1 + \varepsilon} N$ appears close to what current techniques can prove.
\end{abstract}

\thispagestyle{empty}
\newpage\pagenumbering{arabic}
\section{Introduction}
\label{sec:introduction}
The \defn{rigidity} of a $N \times N$ matrix $A$ over a field $\F$ for some rank $r$, denoted by $\R_A^{\F}(r)$, is defined as the minimum number of entries needed to be changed in $A$ in order to reduce the rank of $A$ to at most $r$. 
This concept, introduced by Valiant \cite{valiant1977graphtheoretic}, has broad connections to many topics in complexity theory including circuit lower bounds \cite{valiant1977graphtheoretic,goldreich2016matrix,alman2017probabilistic}, communication complexity \cite{razborov1989rigid,wunderlich2012theorem}, and data structure lower bounds \cite{dvir2019static,natarajanramamoorthy2020equivalence}. In all these connections, researchers have shown that proving a \emph{lower bound} on the rigidity of a matrix implies a breakthrough lower bound for computing that matrix in other models of computation. 

In these applications, the goal is to prove a rigidity lower bound on an \defn{explicit} family of matrices. We say a family of matrices $\BK{A_N}_{N\in \mathbb{N}}$ is explicit, if there is a uniform polynomial-time algorithm that computes all the entries of the matrix $A_N \in \F^{N \times N}$ within time $\poly(N)$. Focusing on explicit matrices is important since lower bounds, for both rigidity and applications, are typically not hard or as interesting to prove for non-explicit matrices. As an example, a counting argument~\cite{valiant1977graphtheoretic} shows that a random $\{-1,1\}$ matrix is very rigid, enough to imply lower bounds via all the known connections, but this is not particularly exciting since counting arguments also show that a random function does not have, for instance, small circuits or efficient communication protocols. 

Moreover, one typically focuses on particular families of explicit matrices which are of interest in the applications. For instance, much work has focused on the rigidity of the Walsh-Hadamard transform, since this is a matrix that is used frequently, both as a linear transformation in circuits, and as the inner product mod 2 function in communication complexity.

\subsection{Rank Parameter Regimes}

Different applications of matrix rigidity focus on different choices of the rank parameter $r$. 
There are three main parameter regimes which have previously been studied. 

The \defn{high-rank regime} where $r \ge N^{\Omega(1)}$ was first considered by Valiant \cite{valiant1977graphtheoretic}, as a tool for proving circuit lower bounds. Valiant proved that, for any matrix family $\BK{A_N}_{N \in \mathbb{N}}$ where $A_N \in \F^{N \times N}$, if there exists a constant $\varepsilon>0$ such that the rigidity of each $A_N$ is lower bounded by $\R^{\F}_{A_N} (N/ \log \log N) \ge N^{1 + \varepsilon}$, then the linear transformation that takes a vector $x \in \F^N$ as an input and outputs $A_N x$ cannot be computed by an arithmetic circuit of depth $O(\log N)$ and size $O(N)$. 
We say the matrix family $\BK{A_N}_{N\in \mathbb{N}}$ is \defn{Valiant rigid}, if it satisfies this rigidity lower bound.
In this high-rank regime, many candidate matrices that were previously conjectured to be Valiant rigid, including some of the most important matrices in this context like the Walsh-Hadamard transform and discrete Fourier transform, were proven to be non-rigid, by a recent line of work \cite{alman2017probabilistic,dvir2019matrix,dvir2020fourier,alman2021kronecker}.

The \defn{mid-rank regime} where $N^{o(1)} > r > \log^{\omega(1)} N$ was first considered by Razborov \cite{razborov1989rigid} (see also \cite{wunderlich2012theorem}), exploring the connection between matrix rigidity and communication complexity of Boolean functions. In the two-party communication complexity model, Alice and Bob want to compute a Boolean function $f: \BK{0,1}^{n} \times \BK{0,1}^n \to \BK{0,1}$ where each party owns $n$ of the $2n$ bits of the input. Many different models of communication have been defined~\cite{goos2018landscapea}, and one typically aims to minimize how much the parties need to communicate with each other to evaluate $f$. Razborov proved that the communication complexity to compute the function $f$ is related to the rigidity of the \emph{communication matrix} $A_f$ of the function $f$. $A_f$ is a $2^n \times 2^n$ matrix whose columns and rows are indexed by strings in $\BK{0,1}^n$, and whose entries are given by $A_f[x,y] = f(x,y)$ for any $(x,y) \in \BK{0,1}^{n} \times \BK{0,1}^n$. Razborov proved that, if the rigidity of each $A_f$ is lower bounded by $\R^{\F}_{A_f} \bk{2^{\log^{\omega(1)} n}} \ge \Omega(4^n)$, then the Boolean function $f$ cannot be computed in $\textsf{PH}^{cc}$, the analogue of polynomial hierarchy in communication complexity. We say the matrix family $\BK{A_N}_{N\in \mathbb{N}}$ is \defn{Razborov rigid}, if it satisfies this rigidity lower bound. Hence, proving rigidity lower bound in the mid-rank regime can imply super strong communication complexity lower bounds.

Despite the importance of rigidity lower bounds in the mid-rank regime, our known rigidity lower bounds for explicit matrices appear quite far from proving that any explicit matrix is Razborov rigid. The best lower bound for an explicit matrix is only $\R^{\F}_{A_f} \bk{2^{\log^{\omega(1)} n}} \ge \Omega(4^n / 2^{\log^{\omega(1)} n})$, which seems to be a big gap from the goal of $\Omega(4^n)$. In this paper, we will give a new reduction from mid-rank to lower rank rigidity, where just a small improvement to our known lower bounds will suffice to achieve Razborov rigidity.

More recently, researchers \cite{alman2021kronecker,alman2023smaller} have started to explore applications of rigidity in the \defn{low-rank regime} where $r = \log^{O(1)} N$. Alman \cite{alman2021kronecker} related the rigidity of a Kronecker power matrix $\kro{A}$ in the {low-rank regime} with the size of \emph{constant-depth} linear circuits that compute the linear transformation $\kro{A}: x \mapsto \kro{A} x$. For instance, \cite{alman2021kronecker} gave the rigidity upper bound $\R^{\F}_{H_4} (1) \le 96$, and used it to obtain a depth-2 linear circuit of size only $O(N^{1.476})$ for the $N \times N$ Walsh-Hadamard transform $H_n$, giving the first improvement on the folklore fast Walsh-Haramard transform algorithm which achieves $O(N^{1.5})$. In this paper, we will prove new, strong rigidity lower bounds in this low-rank regime, proving 
 limits on this approach to designing smaller circuits.

\subsection{Matrices of Interest}

In this paper, we return to the challenge of understanding the rigidity of important families of explicit matrices. 
We focus in particular on two classes of explicit matrix families: the matrix family formed by the Kronecker powers of a small matrix, such as the Walsh-Hadamard transform, and the matrix family formed by a new variant on the Kronecker power we introduce, called the Majority power, including a special family of matrices called the ``distance matrix'', defined as follows:
\begin{itemize}
    \item The $n$-th \defn{Kronecker power} $\kro{A}$ of a matrix $A \in \F^{q \times q}$ is a $q^n \times q^n$ matrix with its columns and rows being indexed by vectors from $[q]^n$,\footnote{We use the notation $[q]$ to represent the set $\BK{0,1,\ldots,q-1}$ for any positive integer $q$.} and entries being $\kro{A}[x, y] \defeq \prod_{i=1}^n A[x_i, y_i]$ for any $x = (x_1, \ldots, x_n) \in [q]^n$ and $y = (y_1, \ldots, y_n) \in [q]^n$.
    \item  In particular, the $n$-th \defn{Walsh-Hadamard matrix} is defined as $H_n \defeq  \bk*{\begin{matrix}
    1 & 1\\
    1 & -1
    \end{matrix}}^{\otimes n} $.
    \item The $n$-th \defn{Majority power} $\maj{A}$ of a matrix $A \in \BK{-1,1}^{q \times q}$ is a $q^n \times q^n$ matrix, with its entry being $\maj{A}[x, y] \defeq \Maj\bk{A[x_1, y_1],\ldots,A[x_n, y_n]}$ for any $x = (x_1, \ldots, x_n) \in [q]^n$ and $y = (y_1, \ldots, y_n) \in [q]^n$, where $\Maj: \BK{-1, 1}^n \to \BK{-1, 1}$ is the majority function.
    \item In particular, the $n$-th \defn{distance matrix} $M_n$ is the $n$-th Majority power of $M_1 \defeq \bk*{\begin{matrix}
    1 & -1\\
    -1 & 1
    \end{matrix}} $. In other words, the entry $M_n[x,y]$ calculates whether the Hamming distance between $x$ and $y$ is greater than $n/2$, for any vectors $x, y \in \BK{0,1}^n$.
\end{itemize}

We focus on these matrix families because of their prominence in applications. For instance, in communication complexity, the Walsh-Hadamard transform is the communication matrix for the ``inner product mod $2$'' function, and Kronecker powers more generally can capture any function which is the parity of constant-sized functions. Similarly, the distance matrix is the communication matrix for testing whether the Hamming distance of the players' inputs is at most a threshold, and the Majority power can capture more general communication problems about thresholds of constant-sized functions \cite{yao2003power,gavinsky2004quantum,huang2006communication}. %

We will prove two types of results for these matrices. First, we will prove new rigidity lower bounds, which for low enough rank are essentially tight and improve on the state of the art for rigidity lower bounds on explicit matrices. Second, we will establish hardness amplification arguments, showing that even a small improvement on our new rigidity lower bound may be challenging: a slight improvement will imply the Razborov rigidity of these matrices, a breakthrough in the fields of matrix rigidity and communication complexity.

\subsection{Previous Rigidity Lower Bounds}
There has been plenty of work on rigidity lower bounds for explicit matrices including Walsh-Hadamard matrices \cite{lokam2001spectral,kashin1998improved,dewolf2006lower}, discrete Fourier transform matrices \cite{shparlinski1999}, Cauchy matrices \cite{shokrollahi1997remark} and generator matrices for linear codes \cite{friedman1993note,pudlak1994combinatorialalgebraic,shokrollahi1997remark}. The best known rigidity lower bound for any explicit matrix over any field is $\R_{A}^{\F_p} (r) \ge \Omega\bk*{\frac{N^2}{r} \cdot \log \frac{N}{r}}$. In particular, for Walsh-Hadamard matrix, the best known rigidity lower bound is $\R_{H_n}^{\F} (r) \ge \frac{N^2}{4r}$, for any field $\F$ which does not have characteristic $2$. (Note that over characteristic $2$, the matrix $H_n$ has rank $1$.) 

\paragraph*{Untouched submatrix argument.}
All the aforementioned lower bounds for explicit matrices are based on a combinatorial approach which we call an \defn{untouched submatrix argument}. The high-level idea of this argument is that, to reduce the rank of a matrix $A$ to $r$, we need to change at least one entry of $A$ in \emph{every} submatrix of $A$ with rank more than $r$. For example, if the matrix $A$ has all of its submatrices being full-rank (such $A$ is called a \defn{totally regular matrix}), then the rigidity of $A$ for rank $r$, $\R_A^{\F}(r)$, is at least the number of entries we need to choose to touch every $r \times r$ submatrix of $A$. We can bound this by applying known results from extremal graph theory (Zarankiewicz problem \cite{kovari1954problem,bollobas2004extremal}), and prove that the rigidity of any totally regular matrix is lower bounded by $\R_A^{\F}(r) \ge \Omega\bk*{\frac{N^2}{r} \cdot \log \frac{N}{r}}$, for rank $r$ with $\Omega(\log N) \le r \le N/2$. There are many known examples of explicit totally regular matrices, such as the Cauchy matrix and generator matrices for some linear codes~\cite{friedman1993note,pudlak1994combinatorialalgebraic,shokrollahi1997remark}.

However, the untouched submatrix argument has a known limitation, and cannot prove a rigidity lower bound better than $\Omega\bk*{\frac{N^2}{r} \cdot \log \frac{N}{r}}$: As shown in \cite{lokam2000rigidity,lokam2009complexity}, there exist totally regular matrices where this lower bound is tight, i.e., $\R_A^{\F}(r) \le O\bk*{\frac{N^2}{r} \cdot \log \frac{N}{r}}$. Hence, in order to prove strong rigidity lower bounds required by the applications (e.g., Valiant rigidity and Razborov rigidity lower bounds), we need to consider new techniques.

\paragraph*{Spectral method.}
In this paper, we will also build on a technique that was previously introduced for bounding a \emph{restricted} version of matrix rigidity called \defn{bounded rigidity}. For a matrix $A$, rank parameter $r$, and positive real number $\theta$, we define $\R_A^{\C}(r, \theta)$ as the minimum number of entries that we need to change in $A$ to reduce the rank to at most $r$, where we are only allowed to change each of those entries by adding or subtracting a value which is \emph{at most $\theta$ in magnitude}. Prior work has given stronger bounds on this bounded rigidity by using spectral tools~\cite{lokam2001spectral,dewolf2006lower,rashtchian2016bounded}. Indeed, if $A$ is approximated by a low-rank matrix $L$ in this bounded sense, then we can control the error in other distances, e.g., the Frobenius distance $\norm{A - L}_{F}$, and then use tools from spectral analysis to further bound this soft version of rigidity.

There are two main caveats of this approach. First, to have the spectral tools and the soft version of rigidity, we need to be in the number field $\C$, hence there is no known rigidity lower bound over a finite field $\F_p$ from this approach. Second, because this only works for bounded rigidity\footnote{Spectral methods have also been used to show that submatrices of $H_n$ or other matrices have high rank, to be used in conjunction with the untouched submatrix method~\cite{kashin1998improved, lokam2001spectral,dewolf2006lower}. However, to our knowledge, spectral methods have only been applied to the sparse error matrix $A-L$ when studying bounded rigidity.}, it does not yield most of the applications of matrix rigidity. In particular, one of the main applications of bounded rigidity was proving ``bounded-coefficient'' arithmetic circuit lower bounds for linear transforms like the Walsh-Hadmard transform. This was seen as a barrier to designing smaller circuits, since all previous circuit constructions were bounded-coefficient. However, a recent line of work~\cite{alman2021kronecker,alman2023smaller} actually showed that such bounded-coefficient lower bounds can be overcome by designing much smaller unbounded-coefficient circuits (by using the low-rank rigidity upper bounds we mentioned earlier!), making bounded-coefficient lower bounds less interesting.

In this paper, we will establish new techniques to get rid of these two caveats, and prove a (regular) rigidity lower bound over $\F_p$ via a new type of spectral method.

\subsection{Boolean Rigidity}
For technical reasons, in this paper, we will focus on a notion of rigidity called \defn{Boolean rigidity}. This notion appears implicitly in many prior works on rigidity and Boolean function analysis, although we're not aware of any that explicitly define it. For each element $x$ in a finite field $\F_p$, we define the \defn{Booleanization} of $x$ as 
    \begin{align*}
        \bool(x) \defeq \begin{cases}
            1 & \text{if } x \equiv 1 \pmod p,\\
            -1 & \text{otherwise,}
        \end{cases}
    \end{align*}
which maps the element $x \in \F_p$ into a Boolean value. The Booleanization of a matrix $A$ is defined as the matrix obtained by taking the Booleanization of each entry of $A$. 
We consider the following variant on rigidity where the low-rank matrix only needs to approximate the given matrix in this Booleanized sense:
\begin{definition}[Boolean rigidity]
    The \defn{Boolean rigidity} of a matrix $A \in \F_p^{N \times N}$ of rank $r$ over $\F_p$, denoted by $\boolRigidity{A}{r}$ is defined as the minimum $s$, such that there is a matrix $L \in \F_p^{N \times N}$ with $\rank(L) \le r$ and a subset of indices $I \subset [N]^2$ with $|I| \ge N^2 - s$, satisfying
    \begin{align*}
        \bool(A[i,j]) = \bool(L[i,j]), \quad \forall (i,j) \in I.
    \end{align*}
\end{definition}

The concept of Boolean rigidity is very closely related to the (regular) rigidity in the following sense: if all the entries of $A$ are $1$ or $-1$ in $\F_p$, then for any rank $r$, we have $$\R_A^{\F_p}(r) \ge \boolRigidity{A}{r} \ge \R_{A}^{\F_p}(O(r^{p-1})).$$ The first inequality is because any rank-$r$ matrix $L$ that is close to $A$ within Hamming distance $s$ must also be close to $A$ within Boolean distance $s$. The second inequality is because if a rank-$r$ matrix $L$ is close to $A$ within Boolean distance $s$, we can construct a matrix $\tilde{L}$ defined as $\tilde{L}[i,j] \defeq \bool\bk{L[i,j]}$ for each entry $(i,j) \in [N]^2$, which is close to $A$ within Hamming distance $s$ as each entry of $A$ is either $1$ or $-1$. By a standard argument using Fermat's little theorem (see \cref{rmk:boolean_close_to_regular} below) we see that $\rank(\tilde{L}) \leq O((\rank(L))^{p-1})$.

Based on this relationship between the Boolean rigidity and the (regular) rigidity, one can see that a Boolean rigidity lower bound which we will prove is stronger than a (regular) rigidity lower bound with the same parameter. Moreover, these two rigidity concepts are equivalent regarding to Razborov rigidity: A $2^n \times 2^n$ matrix $A$ is Razborov rigid, if and only if it satisfies $\boolRigidity{A}{2^{\log^{\omega(1)} n}} \ge \Omega(4^n)$. 
We focus here on Boolean rigidity, even though it is nearly equivalent to (regular) rigidity, because it will simplify the statements of our hardness amplification results below.

The two most closely-related notions of rigidity that have been previously studied are `toggle rigidity' and `sign-rank rigidity'. In toggle rigidity, the low-rank matrix must only have entries among $\{-1,1\}$, whereas in sign rank rigidity, the low-rank matrix must only have the same sign as $A$ in each entry to be considered correct. These notions have previously been considered for both communication complexity and circuit complexity~\cite{razborov1989rigid,wunderlich2012theorem,rashtchian2016bounded,alman2017probabilistic}.

\subsection{Our Results}

In this paper, we give new rigidity lower bounds in the low-rank regime, and new hardness magnification results which show that a substantial improvement to our low-rank lower bounds would give a breakthrough lower bound in communication complexity.

\paragraph*{Nearly tight rigidity lower bound in the low-rank regime.}
Our first result is a generic rigidity lower bound via the spectral method. It shows that the aforementioned limitations of the spectral method in matrix rigidity can be overcome, and spectral arguments can give tight lower bounds for (regular or Boolean) rigidity over a finite field.
\begin{restatable}{theorem}{RigidityLbFromSingularValue}
    \label{thm:rigidity_lb_from_singular_value}
    Let $p$ be a constant prime number and $A \in \BK{-1,1}^{N \times N}$ be a matrix with largest singular value $\sigma_1$ (over $\mathbb{C}$). Then, there is a constant $c >1 $ (depending only on $p$) such that for any rank $r$,
    \begin{align*}
        \boolRigidity{A}{r} \ge N^2 \bk*{\frac{1}{2} - \frac{c^r \sigma_1}{N}}.
    \end{align*}
\end{restatable}

We will see below that \Cref{thm:rigidity_lb_from_singular_value} is tight for many matrices of interest. That said, even generically, we can see that it is nearly tight for small rank $r$ for \emph{any} matrix $A \in \BK{-1,1}^{N \times N}$. Indeed, for any such matrix, its \emph{rank-1} rigidity is at most $N^2/2$, since either the all-$1$ matrix or the all-$(-1)$ matrix will match $A$ in at least half its entries. 

Applying this rigidity lower bound to Kronecker powers and the distance matrix, we obtain the following rigidity lower bounds. (Recall that a Boolean rigidity lower bound also implies the same (regular) rigidity lower bound.)

\begin{restatable}{theorem}{KroneckerLb}
    \label{thm:kronecker_lb}
    Let $p$ be a constant prime number and $q$ be a constant integer. Suppose $A \in \BK{-1,1}^{q \times q}$ is a matrix with $\rank(A) > 1$. Then, there are $c_1 > 0$, $c_2 \in (0,1)$  such that for all positive integers $n$, the Kronecker power $\kro{A}$ has the rigidity lower bound
    \begin{align*}
        \boolRigidity{\kro{A}}{c_1 n} \ge q^{2n} \bk*{\frac{1}{2} - c_2^n}.
    \end{align*}
\end{restatable}

\begin{restatable}{theorem}{HammingLb}
    \label{thm:hamming_lb}
    Let $p$ be a constant prime and $n \in \N$ be a positive integer. 
    For any constant $\eps > 0$, there is a constant $\beta > 0$, such that the distance matrix $M_n$ has the rigidity lower bound
    \begin{align*}
        \boolRigidity{M_n}{\beta \log n} \ge 4^n \bk*{\frac{1}{2} - \frac{1}{n^{1/2 - \eps}}}.
    \end{align*}
\end{restatable}

As the $2^n \times 2^n$ Walsh-Hadamard matrix $H_n$ is the $n$-th Kronecker power of the $2 \times 2$ Walsh-Hadamard matrix $H_1$, \cref{thm:kronecker_lb} gives a nearly tight rigidity lower bound (up to the constant $c_2$) for Walsh-Hadamard matrices for small ranks over finite fields:
\begin{cor} 
\label{cor:walsh_hadamard_lb}
For all constant primes $p$, and all $n$, 
$$4^n\bk*{\frac{1}{2} - 2^{-O(n)}} \le \R_{H_n}^{\F_p} (\Theta(n)) \le \R_{H_n}^{\F_p}(1) \le 4^n \bk*{\frac{1}{2} - 2^{-\Omega(n)}}.$$
\end{cor}
Here, the lower bound is from our \cref{thm:kronecker_lb}, and the upper bound is from prior work on designing constant-depth circuits~\cite{alman2023smaller}. Interestingly, this shows that the rigidity for rank $1$ is hardly smaller than the rigidity for rank $\Theta(n)$ for the $2^n \times 2^n$ Walsh-Hadamard matrix $H_n$. \Cref{thm:hamming_lb} is similarly tight up to the $\varepsilon$ in the exponent, as $\R_{M_n}^{\F_p}(1) \leq 4^n(1/2 - \Theta(1/\sqrt{n}))$ by an upper bound similar to our \Cref{thm:amplification_maj} below.

Note that the lower bound in Corollary~\ref{cor:walsh_hadamard_lb} improves on all the previously known rigidity lower bounds for explicit matrices in the small rank regime: the previous state-of-the-art lower bound 
\cite{friedman1993note,pudlak1994combinatorialalgebraic,shokrollahi1997remark} is $\Omega\bk*{\frac{N^2}{r} \log \frac{N}{r}} = \Omega\bk{N^2}$ when $r = \Theta(\log N)$, where the $\Omega$ hides a small constant. Our Corollary~\ref{cor:walsh_hadamard_lb} improves the leading constant of $N^2$ to the maximum possible $1/2$, and applies in particular to the Walsh-Hadmard transform, for which the previous best lower bound at rank $r = \Theta(\log N)$ was $N^2 / 4r = \Theta(N^2 / \log N)$~\cite{dewolf2006lower}.

As mentioned previously, the low-rank regime has recently been studied extensively in the context of designing smaller low-depth linear circuits. For instance, recent improved circuits for $H_n$ are based on rigidity upper bounds for particular constant sizes: $\R_{H_4}(1) \le 96$~\cite{alman2021kronecker} and $\R_{H_6}(1) \le 1792$~\cite{alman2023smaller}. This is perhaps surprising, as most other applications of matrix rigidity make use of \emph{asymptotic} constructions, rather than constant-sized constructions. As one corollary of our rigidity lower bound, we explain this phenomenon: $H_n$ is too rigid for low ranks when $n$ is super-constant to give improved circuits using the approach of Alman~\cite{alman2021kronecker}, so only constant-sized constructions may give improvements. This motivates searching for techniques other than matrix rigidity upper bounds to get any further improved circuit constructions. See \Cref{sec:obstructiondepth2} for more details.

\paragraph*{Hardness amplification of Boolean rigidity}
Our second, complementary results are new hardness amplification arguments. We give separate arguments for the Kronecker product and the Majority product, as these products yield matrices which are rigid for slightly different parameters. We begin with the Kronecker power.

\begin{restatable}{theorem}{KroneckerAmp}
    \label{thm:amplification_kro}
    Suppose $A \in \BK{-1,1}^{q \times q}$ is a matrix with Boolean rigidity 
      $
      \boolRigidity{A}{r} \le \delta \cdot q^2,
      $
    where $r \le q$ and $\delta\in (0,1/2)$ is a parameter. 
    Assume $A$ has roughly the same number of $1$ and $-1$ entries, i.e., the fraction of $1$'s and the fraction of $-1$'s differ by at most $\alpha \in \bk{0,1}$ in $A$, with $2 \alpha + \delta < {1}/{2}$.
    Then, for any integer parameter $n$, the Boolean rigidity of $\kro{A}$ is at most 
    \begin{align*}
      \boolRigidity{\kro{A}}{2rn} \le q^{2n} \bk*{\frac{1}{2} - \frac{1}{2} \cdot \bk*{\frac{1}{2} - \alpha - \delta}^n}.
    \end{align*}
\end{restatable}

\Cref{thm:amplification_kro} shows that rigidity lower bounds for lower rank parameters can imply rigidity lower bounds for substantially higher rank parameters. To demonstrate, we apply \cref{thm:amplification_kro} to get:

\begin{restatable}{theorem}{StrongerKro}
    \label{thm:stronger_lb_kro_imply_Razborov}
      Let $p$ be a constant prime number, $q$ be a constant integer, and $A \in \BK{-1,1}^{q \times q}$ be a matrix with $\rank(A) > 1$. Suppose we can prove for some constant parameter $\eps > 0 $ that, for infinitely many positive integers $n$,
    \begin{align*}
        \label{ineq:kro_rigidity_lb_for_larger_rank}
        \boolRigidity{\kro{A}}{n^{1 + \eps}} \ge q^{2n} \bk*{\frac{1}{2} - \frac{1}{2^{n/2^{(\log n)^{o(1)}}}}}.\numberthis
    \end{align*}
    Then, the matrix family $\BK{A_n}_{n \in \mathbb{N}}$ is Razborov rigid.
\end{restatable}

In particular, \cref{thm:stronger_lb_kro_imply_Razborov} implies that if we can improve our rigidity lower bound $\boolRigidity{H_n}{\Omega(n)} \ge 4^n \bk*{1/2 - 2^{-O(n)}}$ for the Walsh-Hadamard matrix $H_n$ to work for a slightly larger rank $n^{1 + \eps}$ and a slightly higher error rate $4^n \bk*{1/2 - 2^{-n/2^{\log^{o(1)} n}}}$, it would prove that Walsh-Hadamard matrices are Razborov rigid. This would give a breakthrough lower bound for the communication complexity class $\textsf{PH}^{cc}$.

We emphasize that the higher error rate is \emph{easier} to achieve, but in exchange, the higher rank is \emph{harder} to achieve. More generally, there is a trade-off between the required error and rank parameters one must achieve, depending on how large $n$ is compared to $q$ when one applies \Cref{thm:amplification_kro}. Our \cref{thm:stronger_lb_kro_imply_Razborov} states one example of the trade-off.

This may help to explain why our best rigidity lower bounds for explicit matrices $A$ are stuck at about $\R_A(r) \geq \Omega\bk*{\frac{N^2}{r} \log \frac{N}{r}}$. Currently, for the rank regime $r=2^{\log^{\omega(1)} n}$ needed for Razborov rigidity, there is a sizeable gap between the lower bound $\Omega(N^2 / 2^{\log^{\omega(1)} n})$ we can prove for explicit matrices, and the required lower bound of $\Omega(N^2)$. However, \Cref{thm:stronger_lb_kro_imply_Razborov} gives a setting where there is a small gap between the two: one must only prove a rigidity lower bound for a slightly larger rank as Corollary~\ref{cor:walsh_hadamard_lb}, and it need not even have as low an error parameter. In particular, a large improvement to the best-known rigidity lower bound across all rank regimes may be difficult, but it might be more fruitful to aim for larger improvements focused on specific rank regimes, as we do here in the low-rank regime.

Our second, similar hardness amplification result works for the distance matrix as well as a more general class obtained by taking the Majority power of a small matrix. For technical reasons, our amplification argument needs to start with a matrix with small \emph{probabilistic Boolean rank}, which is 
the ``{worst-case probabilistic}'' version of Boolean rigidity: 
Rather than aiming for a single low-rank matrix which differs from matrix $A$ on at most $\delta$ fraction of entries, we need a \emph{distribution} on low-rank matrices such that each entry of $A$ is correct with error at most $\delta$.
Note that any Boolean function within the communication complexity class $\textsf{PH}^{cc}$ also admits a communication matrix with low probabilistic Boolean rank~\cite{razborov1989rigid,wunderlich2012theorem,alman2017probabilistic,alman2022efficient}, so this suffices for proving that rigidity lower bounds in the low-rank regime would imply a breakthrough communication lower bound:
\begin{restatable}{theorem}{MajorityAmp}
    \label{thm:amplification_maj}
    Let $A \in \BK{-1,1}^{q \times q}$ be a matrix with exactly half its entries equal to $1$. Let $k < n$ be parameters. Suppose the probabilistic Boolean rank of $\maj[k]{A}$ for error rate $\delta \in \bk{0, 1/2}$ is  $r$. Then, the Boolean rigidity of $\maj{A}$ can be bounded as 
    \begin{align*}
      \boolRigidity{\maj{A}}{r} \le q^{2n} \bk*{\frac{1}{2} - \Omega\bk*{(1 - 2\delta) \cdot \frac{\sqrt{k}}{\sqrt{n}}}},
    \end{align*}
  where the $\Omega$ is hiding an absolute constant.
  \end{restatable}

Similar to above, we can apply \Cref{thm:amplification_maj} to get:

\begin{restatable}{theorem}{StrongerMaj}
    \label{thm:stronger_lb_maj_imply_Razborov}
    Let $p$ be a constant prime number, and $\BK{M_n}_{n\in \mathbb{N}}$ be the family of the distance matrices. Suppose we can prove for some constant $\beta > 0$ that, for infinitely many positive integers $n$, 
    \begin{align*}
        \label{ineq:maj_rigidity_lb_for_larger_rank}
        \boolRigidity{M_n}{\beta \log n} \ge 4^n \bk*{\frac{1}{2} - \frac{2^{\bk*{\log \log n}^{o(1)}}}{n^{1/2}}}. \numberthis
    \end{align*}   
    Then, the matrix family $\BK{M_n}_{n \in \mathbb{N}}$ is Razborov rigid.
\end{restatable}

This again shows that a slight improvement to our rigidity lower bound in \Cref{thm:hamming_lb} would yield a breakthrough communication lower bound. Notably, both the rigidity lower bound of \Cref{thm:hamming_lb}, and the required lower bound of \Cref{thm:stronger_lb_maj_imply_Razborov}, focus on the same, lower error parameter (of about $1/2 - 1/\sqrt{n}$) compared to \Cref{thm:stronger_lb_kro_imply_Razborov} (which focuses on about $1/2 - 1/2^{\Theta(n)}$).

\subsection{Other Related Works}

Matrix rigidity has been studied extensively since its introduction by Valiant~\cite{valiant1977graphtheoretic}. We refer the reader to surveys by Lokam~\cite{lokam2009complexity}, Ramya~\cite{ramya2020recent}, and Golovnev~\cite{golovnev2020matrix} about known upper and lower bounds on rigidity, and connections with other areas of computer science. We briefly expand below on a few of the works most related to ours.

\paragraph{Rigidity upper bounds.}
Alman and Williams~\cite{alman2017probabilistic} showed that the Walsh-Hadamard transform is not Valiant rigid. Letting $N = 2^n$ denote the side-length of $H_n$, they showed that for every $\varepsilon>0$, $\R_{H_n}(N^{1 - \Theta(\varepsilon^2 / \log ( 1/\varepsilon))}) \leq N^{1 + \varepsilon}.$
They also showed an upper bound for lower rank but higher sparsity: for every parameter $s$, they showed
$\R_{H_n}((n/\log s)^{O(\sqrt{n \log s})}) \leq N^2 / s.$
These two upper bounds hold over any field. Later work by Alman~\cite{alman2021kronecker} showed that they also hold for $M_n$, and more generally any Kronecker power or Majority power matrix family.
A folklore upper bound uses the singular values of $H_n$ to prove that $H_n$ is not very rigid for very high rank: $\R^{\mathbb{C}}_{H_n}(N/2) \leq N.$
Finally, Alman, Guan and Padaki~\cite{alman2023smaller} gave an upper bound on the rank-1 rigidity of $H_n$ as part of their construction of low-depth arithmetic circuits:
$\R_{H_n}(1) \leq N^2 \cdot \left( \frac12 - \frac{\Theta(1)}{\sqrt{N}} \right).$

\paragraph{Semi-explicit rigid matrices.}

A recent line of work~\cite{alman2022efficient,bhangale2024rigid,chen2020almost,chen2021inverse,huang2021averagecase}, gave a ``somewhat-explicit'' construction of a matrix family which is Razborov rigid. They proved that there is a $\mathsf{P^{NP}}$ machine, which, on input $1^n$, outputs an $n \times n$ matrix $A_n$ over $\mathbb{F}_2$ such that 
$\R_{A_n}(2^{\Omega(\log N / \log \log N)}) \geq N^2 \cdot \left(  \frac12 - o(1) \right).$ This is the only other line of work we're aware of that gives rigidity lower bounds with maximal sparsity parameters $N^2 \cdot \left(  \frac12 - o(1) \right)$ for non-random matrices.

\section{Technical Overview}
\label{sec:tec_overview}

\subsection{Rigidity Lower Bound via Spectral Method}

Our lower bound is built on the spectral method that was previously used to bound the bounded rigidity $\R_A^{\C}(r, \theta)$. To explain the high-level idea of this method, suppose the matrix $A$ is approximated by a low-rank matrix $L$ in the bounded sense; that is, $L$ differs from $A$ in at most $s$ entries, with the magnitude of the difference in each differing entry pair bounded by $\theta$. Then, the Hamming distance $s$ between $L$ and $A$ can be lower bounded by the Frobenius distance as $\norm{A -L}_F^2 \le \theta^2 s$, which can be further analyzed using spectral tools.

There are three technical challenges of using the spectral method to prove a nearly tight regular or Boolean rigidity lower bound $\boolRigidity{A}{r} \ge N^2 \cdot (1/2 - o(1))$ over a finite field $\F_p$:
\begin{itemize}
  \item First, the Frobenius distance between matrices is defined only over the field $\C$.
  \item Second, to make the Frobenius distance upper bounded by the Hamming distance, we crucially rely on the fact that the difference between $A$ and $L$ is entry-wisely bounded by $\theta$ in magnitude. In particular, since $\norm{A -L}_F^2$ scales with both $s$ and $\theta$, even small constant factor changes in $\theta$ can make it difficult to achieve the leading constant of $1/2$ in the sparsity $s$ that we aim to achieve in our lower bound.
  In particular, to get the exact coefficient $1/2$ in prior work, we require $\theta$ to be at most $2$, which is a strong requirement.
  \item Finally, to get a lower bound on the Frobenius distance between $A$ and a rank-$r$ matrix $L$, prior work needs to use that the low-rank decomposition $U, V$ of $L$, i.e., the two $r \times n$ matrices such that $L = U^{\top}V$, have each entry bounded as well. This does not immediately follow, even from bounded rigidity upper bounds, but this additional property was obtained via a tool from convex geometry called John's Theorem \cite{john2014extremum} in the previous state-of-the-art work \cite{rashtchian2016bounded}. (We will ultimately replace this with a \emph{simpler} bound in the finite field setting.)
\end{itemize}

In this work, we show, surprisingly, that we can address these three challenges at the same time by a new, powerful generic transformation from a low-rank approximation over $\F_p$ to a low-rank approximation over $\C$. Specifically, for a matrix $A \in \BK{-1,1}^{N \times N}$, if it can be approximated by a rank-$r$ matrix $L$ over $\F_p$, then we construct another matrix $\tilde{L} \in \BK{-1,1}^{N \times N}$ with rank $\tilde{r}$ over $\C$ that approximates $A$ within the same accuracy as $L$.

As each entry of $\tilde{L}$ is either $1$ or $-1$, $\tilde{L}$ approximates $A$ in the bounded sense with $\theta = 2$, which will allow us to make use of some tools from the spectral approach of prior work. Moreover, we prove that the matrix $\tilde{L}$ admits a natural low-rank decomposition $U, V \in \C^{r\times n}$ such that each entry of them is bounded by $2^{O(r)}$. Defining $\tilde{L}$ and proving this requires a careful polynomial interpolation in conjunction with other algebraic manipulations. With all these properties, we can apply a variant on the spectral method on $\tilde{L}$ to lower bound the error by $N^2 (1/2 - o(1))$.

\subsection{Hardness Amplifications for Kronecker Product}
\label{subsec:overview_amplification_kro}
In this subsection, we will overview our techniques for hardness amplification. 
Suppose we are given a rank-$r$ matrix $L$ that approximates the matrix $A \in \BK{-1, 1}^{q \times q}$ within a constant error rate $\delta$. Our goal is to construct a matrix $\tilde{L} \in \F_p^{q^n \times q^n}$ from $L$ without a large blow up in the rank, such that $\tilde{L}$ approximates $\kro{A}$ with an accuracy $(1/2 + 2^{-O(n)})$, meaning that $\tilde{L}$ agree with $\kro{A}$ on $(1/2 + 2^{-O(n)})$-fraction of entries.

The most natural way for constructing $\tilde{L}$ is to consider $\kro{L}$. As $L$ is a low-rank approximation of ${A}$, we can think of $\kro{L}$ as approximating $\kro{A}$ within certain accuracy. Following some careful calculation, one can check $\kro{L}$ does approximate $\kro{A}$ with an accuracy $(1/2 + 2^{-O(n)})$. The only issue is that $\kro{L}$ has a much higher rank $r^n$ than what we want.

In this paper, we construct a matrix $\tilde{L}$ with a much lower rank $2nr$, without harming the asymptotic accuracy $(1/2 + 2^{-O(n)})$.
To define the low-rank approximation $\tilde{L}$ of $\kro{A}$, we consider $\kro{L}$ as an intermediate matrix and intuitively think of giving the approximation in two steps:
\begin{itemize}
    \item $\kro{L}$ approximates $\kro{A}$ with an accuracy $(1/2 + 2^{-O(n)})$, although the rank of $\kro{L}$ is high.
    \item Let $\pi: \F_p^n \to \F_p$ be the multiplication function, i.e., $\pi(z_1, \ldots, z_n) \defeq z_1 \cdots z_n$. Then, each entry of $\kro{L}$ can be represented as
        $\kro{L}[x,y] = \pi\bk*{L[x_1, y_1], \ldots, L[x_n,y_n]}.$
     We can further use a degree-$1$ polynomial $\tilde{\pi}$ to approximate the high-degree function $\pi$. Inspired by the fact that when $z_1, \ldots, z_n \in \BK{-1,1}$, the function $\pi(z_1, \ldots, z_n)$ counts the parity of the number of $-1$'s among its inputs, we construct $\tilde{\pi}$ using a known construction of probabilistic polynomials for the PARITY function~\cite{razborov1987lower, smolensky1987algebraic}. Using it, we define a low-rank matrix $\tilde{L}$ by 
    $
        \tilde{L}[x,y] = \tilde{\pi}\bk*{L[x_1, y_1], \ldots, L[x_n,y_n]}.
    $
    This low-rank matrix $\tilde{L}$ approximates $\kro{L}$ within a certain accuracy.
\end{itemize}
Now as $\kro{L}$ approximates $\kro{A}$, and $\tilde{L}$ further approximates $\kro{L}$ in some sense, it would be natural to use a union bound to conclude that $\tilde{L}$ approximates $\kro{A}$ well. However, both parts of the approximations are actually very weak: for example, for the second part of the approximation, the best degree-1-polynomial approximation $\tilde{\pi}$ for $\pi$ can only achieve an accuracy of $(1/2 + o(1))$, meaning that $\tilde{L}$ only agrees with $\kro{L}$ on a $(1/2 + o(1))$-fraction of entries. We cannot use a union bound to compose such two weak approximations. 
Fortunately, we can address this issue by carefully considering these two approximations in terms of \emph{Boolean rigidity}: When both the approximations fail at some entry $(x,y)$, their composition yields a correct approximation at $(x,y)$, i.e., $\bool(\tilde{L}[x,y]) = -\bool(\kro{L}[x,y])= \bool (\kro{A}[x,y])$. A careful calculation will show that $\bool(\tilde{L})$ approximates $\bool(\kro{A})$ within accuracy $(1/2 + 2^{-O(n)})$, relying on the fact that the errors in these two approximations are essentially independent.

\subsection{Hardness Amplification for Majority Product}
The Kronecker product has the property that $\kro{A} = \kro[n/k]{(\kro[k]{A})}$, and we implicitly made use of this above to reduce from $\kro{A}$ to $\kro[k]{A}$. Here, we also need to reduce from $\maj{A}$ to $\maj[k]{A}$, but unfortunately, in general, $\maj[n/k]{(\maj[k]{A})} \neq \maj{A}$. We begin by calculating the number of entries in which these differ, which contributes to our final error.

Suppose $A \in \BK{-1, 1}^{q \times q}$ and we are given, via a rigidity upper bound, a rank-$r$ matrix $L$ that approximates the matrix $\maj[k]{A} \in \BK{-1, 1}^{q^k \times q^k}$ within a constant error rate $\delta$. Our goal is to construct a matrix $\tilde{L}$ with a similar rank as $L$ that approximates $\maj{A}$ with an accuracy $(1/2 + \Omega\bk{\sqrt{k/n}})$. 

We follow the same framework of constructing $\tilde{L}$ as for the Kronecker powers: 
\begin{itemize}
  \item As $L$ approximates $\maj[k]{A}$ well, and $\maj{A}$ is similar to the $(n/k)$-th majority power of $\maj[k]{A}$ (as calculated earlier), we use the relatively high-rank matrix $\maj[(n/k)]{L}$ to approximate $\maj{A}$. 
  \item Then, we use a degree-$1$ polynomial $\tilde\Maj$ to approximate the majority function $\Maj\bk{z_1, \ldots, z_{n/k}}$ and define the low-rank matrix $\tilde{L}$ that approximates $\maj[(n/k)]{L}$ correspondingly. In the case of majority product, the degree-$1$ polynomial can be simply chosen as $\tilde{\Maj}(z_1, \ldots, z_{n/k}) = z_1$. One can calculate that this simple dictator function will approximate the majority function with an accuracy $(1/2 + \Omega\bk{\sqrt{k/n}})$.
\end{itemize}

Similar to before, we may naturally aim to use Boolean rigidity to compose these two approximations. However, this is unfortunately impossible: the composed approximation, that uses $\tilde{L}$ to approximate $\maj{A}$, is not guaranteed to be have an accuracy $(1/2 + \Omega\bk{\sqrt{k/n}})$.
To see why this is the case, we consider another intermediate matrix $\tilde{A} \in \BK{-1,1}^{q^n \times q^n}$, which is the analogue of $\tilde{L}$ but with $A$ in place with $L$ in its definition. Specifically, 
\begin{align*}
  \tilde{A}[x,y] \defeq \maj[k]{A}\bk{x_{[1,k]}, y_{[1,k]}} = \Maj\bk{A[x_1, y_1], \ldots, A[x_k, y_k]}, \quad \forall x, y \in [q]^n,
\end{align*}
where $x_{[1,k]}$ represents a prefix of length $k$ in $x = (x_1, \ldots, x_n)$.
The errors in the whole approximation can be classified into two types: the inherent error incurred by the given low-rank matrix $L$ (i.e., the error between $\tilde{L}$ and $\tilde{A}$), and the approximation error (i.e., the error between $\tilde{A}$ and $\maj{A}$). 
The previous composition argument for the Kronecker product crucially relies on the fact that the two types of errors are independent, which is not true for the majority product, as the matrix $L$ only approximates $\maj[k]{A}$ in a coarse sense. For example, if we consider indices $(x,y) \in [q]^n \times [q]^n$ sampled uniformly at random, then:
\begin{itemize}
  \item The random variable $\indicator{\tilde{A}[x,y] = \maj{A}[x,y]}$ is correlated with the random variable $S \defeq A[x_1, y_1] + \cdots + A[x_k, y_k]$. When $S > 0$ (meaning that $\tilde{A}[x,y] = 1$), the larger $S$ is, the more likely that $\maj{A}[x,y] = 1$.
  \item The random variable $\indicator{\tilde{A}[x,y] = \tilde{L}[x,y]} = \indicator{\maj[k]{A}[x_{[1,k]}, y_{[1,k]}] = L[x_{[1,k]}, y_{[1,k]}]}$ can also possibly be correlated with the random variable $S$. This is because the given matrix $L$ is only guaranteed to approximate $\maj[k]{A}$ \emph{on average}, so the error can be distributed arbitrarily.
\end{itemize}
Thus, the two errors given by $\indicator{\tilde{A}[x,y] \neq \maj{A}[x,y]}$ and $\indicator{\tilde{A}[x,y] \neq \tilde{L}[x,y]}$ can be correlated, making it hard to compose them. To address this issue, we need to assume the stronger condition that the \emph{probabilistic Boolean rank} with error rate $\delta$ of $\maj[k]{A}$ is at most $r$. This means the errors in $L$ are now uniformly distributed, allowing the two errors mentioned above to be independent.

\section{Preliminaries}
\label{sec:preliminaries}

\subsection{Matrix Rigidity}

\begin{definition}[Matrix rigidity]
    The \defn{rigidity} of a matrix $A \in \F^{N \times N}$ of rank $r$ over a field $\F$, denoted by $\R_{A}^{\F}(r)$, is defined as the minimum $s$, such that there is a rank-$r$ matrix $L \in \F^{N \times N}$ that agree with $A$ in at least $N^2 - s$ entries, i.e., there is subset of indices $I \subset [N]^2$ with $|I| \ge N^2 - s$, satisfying
    \begin{align*}
        A[i,j] = L[i,j], \quad \forall (i,j) \in I.
    \end{align*}
\end{definition}

\begin{definition}[Razborov rigidity]
    We say a matrix family $\BK{A_n}_{n \in \mathbb{N}}$ is \defn{Razborov rigid} over some field $\F$, if there is a constant $\delta > 0$, such that for any $2^n \times 2^n$ matrix $A_n$, we have \[\R_{A_n}^{\F} \bk*{2^{\log^{\omega(1)} n}} \ge \delta \cdot 4^n.\]
\end{definition}

\subsection{Rigidity Variants}

\begin{definition}[Boolean rigidity]
    For any $x \in \F_p$, we define its \defn{Booleanization} as 
    \begin{align*}
        \bool(x) \defeq \begin{cases}
            1 & \text{if } x \equiv 1 \pmod p,\\
            -1 & \text{otherwise.}
        \end{cases}
    \end{align*}
    The \defn{Boolean rigidity} of a matrix $A \in \F_p^{N \times N}$ of rank $r$ over $\F_p$, denoted by $\boolRigidity{A}{r}$ is defined as the minimum $s$, such that there is a matrix $L \in \F_p^{N \times N}$ with $\rank(L) \le r$ and a subset of indices $I \subset [N]^2$ with $|I| \ge N^2 - s$, satisfying
    \begin{align*}
        \bool(A[i,j]) = \bool(L[i,j]), \quad \forall (i,j) \in I.
    \end{align*}
\end{definition}

\begin{remark}
\label{rmk:boolean_close_to_regular}
    For any rank $r$, we have $\R_{A}^{\F_p} (O(r^{p-1})) \le \boolRigidity{A}{r} \le \R_{A}^{\F_p} (r)$.
\end{remark}

\begin{proofof}{the first inequality}
    Let $L \in \F_p^{N \times N}$ be the rank-$r$ matrix that approximates $A$ in the boolean sense, then $\tilde{L} \in \F_p^{N \times N}$ defined by $\tilde{L}[i,j] \defeq \bool\bk{L[i,j]}$ for any $(i,j) \in [N]^2$ approximates $A$ in the regular sense with the same error rate. Moreover, we prove that $\rank\bk{\tilde{L}} \le O(r^{p-1})$ below, using a standard technique called the \emph{polynomial method}.

    Let $L \eqdef U^{\top} V$ be the low-rank decomposition of $L$, where $U,V$ are matrices in $\F_p^{r \times N}$. Then, the $(i,j)$-entry of $L$ can be represented as 
    $L[i,j] = \vec{u_i}^\top \vec{v_j} = \sum_{k=1}^{r} \vec{u_i}[k] \vec{v_j}[k]$, where $\vec{u_i}$ is the $i$-th column vector of $U$ and $\vec{v_j}$ is defined similarly.
    By Fermart's little theorem, each entry of $\tilde{L}$ can be represented as 
    \begin{align*}
        \label{eq:low-rank_from_Fermart}
        \tilde{L}[i,j] = \bool\bk{L[i,j]}
        = 1 - 2 \cdot \bk{L[i,j] - 1}^{p-1}
        = 1 - 2 \cdot \bk*{\sum_{k=1}^{r} \vec{u_i}[k] \vec{v_j}[k] - 1}^{p-1}. \numberthis
    \end{align*}
    By expanding the $(p-1)$-th power in \eqref{eq:low-rank_from_Fermart}, we can further write the RHS of \eqref{eq:low-rank_from_Fermart} as 
    \begin{align*}
        \label{eq:expand_power}
        \tilde{L}[i,j] = \sum_{\vec{\alpha} \in \BK{0,1}^r,\, |\vec{\alpha}|\le p-1} C_{\vec{\alpha}} \vec{u_i}^{\vec{\alpha}} \vec{v_j}^{\vec{\alpha}}, \numberthis
    \end{align*}
    where $C_{\vec{\alpha}}$ are constant coefficients,  and we use the notation that, for an exponent vector $\vec{\alpha}$ and a vector $\vec{u}_i$, we write $\vec{u}_i^{\vec{\alpha}}$ to denote $u_{i1}^{\alpha_1} u_{i2}^{\alpha_2} \cdots u_{ir}^{\alpha_r}$. As there are at most $O(r^{p-1})$ different terms in the sum in \eqref{eq:expand_power}, one can conclude that $\rank\bk{\tilde{L}} \le O(r^{p-1})$ and $\R_{A}^{\F_p} (O(r^{p-1})) \le \boolRigidity{A}{r}$. 
\end{proofof}

We also consider the notion of \emph{probabilistic} Boolean rank, which combines Boolean rank with the notion of probabilistic rank from prior work~\cite{alman2017probabilistic}:

\begin{definition}[Probabilistic Boolean rank]
    For a matrix $A \in \F^{N \times N}$, The \defn{probabilistic Boolean rank} of an error rate $\delta \in (0,1)$ is defined as the minimum rank $r$, such that there is a distribution $\mathcal{L}$ of $N \times N$ rank-$r$ matrices, satisfying
    \begin{align*}
        \Pr_{L \sim \mathcal{L}}\Bk*{\bool\bk{A[i,j]} \neq \bool\bk{L[i,j]}} \le \delta, \quad \forall (i,j) \in [N]^2.
    \end{align*}
\end{definition}

\subsection{Matrix Power Families}

\begin{definition}[Kronecker power]
    For any matrix $A \in \BK{-1,1}^{q \times q}$, its Kronecker $n$-th power $\kro{A}$ is a matrix in $\BK{-1, 1}^{q^n \times q^n}$, with its columns and rows being indexed by vectors from $[q]^n$, and entries being
    \begin{align*}
        \kro{A}[x, y] \defeq \prod_{i=1}^n A[x_i, y_i],
    \end{align*}
    for any $x = (x_1, \ldots, x_n) \in [q]^n$ and $y = (y_1, \ldots, y_n) \in [q]^n$.
\end{definition}

\begin{definition}[Majority power]
    For any matrix $A \in \BK{-1,1}^{q \times q}$, its $n$-th power $\maj{A}$ is a matrix in $\BK{-1, 1}^{q^n \times q^n}$, with its columns and rows being indexed by vectors from $[q]^n$, and entries being
    \begin{align*}
        \maj{A}[x, y] \defeq \Maj\bk{ A[x_1, y_1], \ldots, A[x_n, y_n]},
    \end{align*}
    for any $x = (x_1, \ldots, x_n) \in [q]^n$ and $y = (y_1, \ldots, y_n) \in [q]^n$, where $\Maj: \BK{-1, 1}^n \to \BK{-1, 1}$ is the majority function defined as 
    \begin{align*}
        \Maj(a_1, \ldots, a_n) = \begin{cases}
            1 & \text{if } a_1 + \cdots + a_n \ge 0\\
            -1 & \text{otherwise}.
        \end{cases}
    \end{align*}
\end{definition}

\begin{definition}[Walsh-Hadamard matrix]
    We define the $n$-th \defn{Walsh-Hadamard matrix} $H_n \defeq \kro{H_1}$ as the $n$-th Kronecker power of the matrix $H_1 \defeq \bk*{\begin{matrix}
        1 & 1\\
        1 & -1
        \end{matrix}} $.
\end{definition}

\begin{definition}[Distance matrix]
    We define the $n$-th \defn{distance matrix} $M_n$ as the $n$-th Majority power of $M_1 \defeq \bk*{\begin{matrix}
        1 & -1\\
        -1 & 1
        \end{matrix}} $. 
    Equivalently, the distance matrix $M_n$ can be defined entry-wisely by 
    \begin{align*}
        M_n[x,y] = 2\indicator{|{x-y}| \le n/2}-1, \quad \forall x,y \in \BK{-1,1}^n.
    \end{align*}

\end{definition}

\section{Nearly Tight Rigidity Lower bounds}
\label{sec:lower_bound}
\newcommand{\UinC}{\tilde{U}}
\newcommand{\VinC}{\tilde{V}}
\newcommand{\rinC}{\tilde{r}}
\newcommand{\LinC}{\tilde{L}}
\newcommand{\vecUinCsub}{\tilde{\vec{u}_i}^{(\text{sub})}}
\newcommand{\vecVinCsub}{\tilde{\vec{v}_j}^{(\text{sub})}}
\newcommand{\vecUinC}{\tilde{\vec{u}_i}}
\newcommand{\vecVinC}{\tilde{\vec{v}_j}}
\newcommand{\vecUinCrow}{\tilde{\vec{u}}_{(k)}}
\newcommand{\vecVinCrow}{\tilde{\vec{v}}_{(k)}}

In this section, we will prove \cref{thm:rigidity_lb_from_singular_value} which establishes a rigidity lower bound for any binary matrix in a finite field in terms of its largest singular value. After that, we will apply \cref{thm:rigidity_lb_from_singular_value} to get nearly tight rigidity lower bounds for two types of matrices in small ranks.

\subsection{Rigidity Lower Bounds by Largest Singular Value}
In this subsection, we prove \cref{thm:rigidity_lb_from_singular_value}. 

\RigidityLbFromSingularValue*

The key idea behind the proof is to consider the matrix $A$ over the number field $\mathbb{C}$, even though we are arguing about its (Boolean) rigidity over $\F_p$. Specifically, the proof consists of two parts:
\begin{itemize}
    \item Showing that if a binary matrix can be approximated well by some low-rank matrix in the \emph{finite field $\F_p$}, then it can also be approximated well by some binary ``low-rank'' matrix in the \emph{number field $\mathbb{C}$} (for a more restricted notion than low-rank which we will define shortly, where the low-rank decomposition must use small entries).
    \item Bounding the ``rigidity'' (using this restricted rank notion) of binary matrices in $\mathbb{C}$ in terms of their largest singular values.
\end{itemize}

Before getting to the full proof, we give the key technical lemma behind the first step:

\begin{lemma}
\label{lm:transformation_of_low-rank_matrix_to_C}
    Suppose $L\in \F_p^{N \times N}$ is a matrix of rank $r$ over $\F_p$. Then, the matrix $\LinC\in \BK{-1,1}^{N \times N}$, defined as the entry-wise Booleanization of $L$, i.e., 
    \[\LinC[i,j] = \bool(L[i,j]), \quad \forall (i,j)\in [N]^2,\]
    satisfies the following properties:
    \begin{itemize}
        \item Viewed as a matrix over $\C$, the rank of $\LinC$ over $\C$ is at most $\rinC \defeq (p^3 + 1)^r$.
        \item Moreover, $\LinC$ admits a low-rank decomposition $\LinC = \UinC^\top \VinC$ over $\C$, where $\UinC, \VinC \in \mathbb{C}^{\rinC \times N}$ and each entry of $\UinC$ and $\VinC$ has magnitude at most $C^r$ for some constant $C$ which depends only on $p$.
    \end{itemize}
\end{lemma}

\begin{proof}
    Since $L$ is a rank-$r$ matrix over $\F_p$, it admits a low-rank decomposition $L \overset{\F_p}{=} U^{\top} V$, where $U, V \in \F_p^{r \times N}$. We can interpret $U$, $V$ as integer matrices from $\Z^{r \times N}$ whose entries are integers between $0$ and $p-1$, and express the decomposition as $L \equiv U^{\top} V \pmod{p}$.

    Let $\omega \defeq e^{2\pi i/p}$ be a $p$-th root of unity. By polynomial interpolation, there is a polynomial $f \in \mathbb{C}[x]$ with degree at most $p$, such that $f(\omega^{k}) = \bool \bk{k}$ for all $k \in [p]$. Then, the entries of $\LinC$ can be written as
        \begin{align*}
        \LinC[i,j] = \bool(L[i,j]) = f\bk*{\omega^{L[i,j]}} = f\bk*{\omega^{(U^{\top} V)[i,j]}}, \quad \forall (i,j) \in [N]^2.  \numberthis \label{ineq:lb_LinC_entries_def}
    \end{align*}
    We will use \eqref{ineq:lb_LinC_entries_def} to construct the desired low-rank decomposition of $\LinC$ below.

    Let $\vec{u}_i$ be the $i$-th column vector of $U$ for any $i \in [N]$, and $\vec{v}_j$ similarly be the $j$-th column vector of $V$. Write $\vec{u}_i = \bk{u_{i1}, \ldots, u_{ir}}^{\top}$ and $\vec{v}_j = \bk{v_{j1}, \ldots, v_{jr}}^{\top}$. We will expand $\LinC[i,j]$ as a sparse polynomial in  $u_{i1}v_{j1}, \ldots, u_{ir}v_{jr}$, namely
    \begin{align*}
        \LinC[i,j] = F\bk*{u_{i1}v_{j1}, \ldots, u_{ir}v_{jr}},  \numberthis \label{eq:L=F(uv) in sec_rigidity_lb}
    \end{align*}
    where $F$ is a $r$-variable polynomial which we construct next, which has degree at most $p^{3}$ in each variable.

    By polynomial interpolation, there is a polynomial $g \in \mathbb{C}[x]$ with degree at most $p^2$, such that $g(k) = \omega^{k}$, for all $k \in \Bk[big]{p^2}$.
    Expand $g$ into a sum of monomials as $g(x) = c_0 + c_1 x + \cdots + c_{p^2} x^{p^2}$. Then,
    \begin{align*}
        \omega^{(U^{\top} V)[i,j]} = {\prod_{k=1}^{r} \omega^{u_{ik} v_{jk}}}
        = {\prod_{k=1}^r g(u_{ik} v_{ik})}
        = {\prod_{k=1}^r \bk*{c_0 + c_1 \bk{u_{ik} v_{ik}} + \cdots + c_{p^2} \bk{u_{ik} v_{ik}}^{p^2}}},
    \end{align*}
    meaning that there is a $r$-variable polynomial $G$ of degree at most $p^{2}$ in each variable, such that 
    $\omega^{(U^{\top} V)[i,j]} = G\bk*{u_{i1}v_{j1}, \ldots, u_{ir}v_{jr}}$. Then, as $f$ is a polynomial of degree at most $p$, 
    \begin{align*}
        \LinC[i,j] = f\bk*{\omega^{(U^{\top} V)[i,j]}} = f\bk[BBig]{G\bk*{u_{i1}v_{j1}, \ldots, u_{ir}v_{jr}}}
    \end{align*}
    can be further expressed as $F(u_{i1}v_{j1}, \ldots, u_{ir}v_{jr})$, where $F \defeq f \circ G$ is a polynomial of degree at most $p^{3}$ in each variable. Moreover, expanding $F$ as
    \begin{align*}
        F(u_{i1}v_{j1}, \ldots, u_{ir}v_{jr}) = \sum_{\vec{\alpha} \in \Bk{p^3 + 1}^{r}} C_{\vec{\alpha}} \vec{u}_i^{\vec{\alpha}} \vec{v}_i^{\vec{\alpha}},
    \end{align*}
    we have further that each coefficient $C_{\vec{\alpha}}$ is bounded in magnitude by $C^{r}$ for some constant $C$ (depending on $p$). This is because both $f$ and $g$ are polynomials that only depend on $p$, with their $\l_1$-norm (sum of magnitudes of all the coefficients) being bounded by constants $C_f, C_g$, respectively, that only depend on $p$. Then, the $\l_1$-norm of $G$ is bounded by $\bk{C_g}^r$ and the $\l_1$-norm of $F$ is bounded by $C_f \cdot \bk{C_g}^{rp}$, which means that each coefficient $C_{\vec{\alpha}}$ of $F$ is bounded by $C_f \cdot \bk{C_g}^{rp} \eqdef C^r$ for some constant $C$ that only depends on $p$.

    With the low-degree polynomial expression \eqref{eq:L=F(uv) in sec_rigidity_lb} of $\LinC$, we 
    can construct matrices $\UinC, \VinC \in \mathbb{C}^{\rinC \times N}$ such that $\LinC = \UinC^{\top} \VinC$, as follows:
    \begin{itemize}
        \item Let $\vecUinC$ be the $i$-th column vector of $\UinC$ for any $i \in [N]$ and $\vecVinC$ be the $j$-th column vector of $\VinC$. Both $\vecUinC$ and $\vecVinC$ are $\bk{p^3 + 1}^r$-dimensional vectors, with their coordinates indexed by all possible $\vec{\alpha} \in \Bk{p^3 + 1}^r$.
        \item In $\vecUinC$, the entry indexed by $\vec{\alpha}$ is $\vecUinC[\vec{\alpha}] \defeq C_{\vec{\alpha}} \vec{u}_i^{\vec{\alpha}}$.
        \item In $\vecVinC$, the entry indexed by $\vec{\alpha}$ is $\vecVinC[\vec{\alpha}] \defeq \vec{v}_j^{\vec{\alpha}}$.
    \end{itemize}
    Then, we can check $\LinC = \UinC^{\top} \VinC$ by 
    \begin{align*}
        \UinC^{\top} \VinC[i,j] = \angbk{\vecUinC, \vecVinC} = \sum_{\vec{\alpha} \in \Bk{p^3 + 1}^{r}} C_{\vec{\alpha}} \vec{u}_i^{\vec{\alpha}} \vec{v}_i^{\vec{\alpha}} = \LinC[i,j], \quad \forall (i,j) \in [N]^2.
    \end{align*}
    Hence, the matrix $\LinC$ has rank at most $\rinC = (p^3+1)^r$ over $\mathbb{C}$, with a low-rank decomposition $\LinC = \UinC^{\top} \VinC$ such that each entry of $\UinC$ and $\VinC$ is bounded in magnitude by $C^{r}$ for some constant $C$ which depends only on $p$.
    \end{proof}

We now prove \cref{thm:rigidity_lb_from_singular_value} using \cref{lm:transformation_of_low-rank_matrix_to_C}.
\begin{proofof}{\cref{thm:rigidity_lb_from_singular_value}}
    Suppose that $\boolRigidity{A}{r} = s$, then by the definition of matrix rigidity, 
    there is a rank-$r$ matrix $L$ in $\F_p$, such that 
    \begin{align*}
        \bool\bk{A[i,j]} = \bool\bk{L[i,j]}, \quad \forall (i,j) \in I, \numberthis \label{eq:A=UV mod p in sec_rigidity_lb}
    \end{align*}
    for some set $I \subset [N]^2$ of indices with $\abs{I} \ge N^2 - s.$ 
    By applying \cref{lm:transformation_of_low-rank_matrix_to_C} on $L$, the matrix $\LinC$ defined by applying Booleanization entry-wise on $L$ admits a low-rank decomposition $\LinC = \UinC^\top \VinC$ over $\C$, where $\UinC, \VinC \in \C^{\rinC \times N}$ with each entry being bounded in magnitude by $C^r$. Hence, \eqref{eq:A=UV mod p in sec_rigidity_lb} can be written as 
    \begin{align*}
        A[i,j] = \bool\bk{A[i,j]} = \bool\bk{L[i,j]} = \LinC[i,j] = \bk*{\UinC^\top \VinC}[i,j], \quad \forall (i,j)\in I.\numberthis \label{eq:lb_wirte_rigidity_over_C}
    \end{align*}
    
    Now we use \eqref{eq:lb_wirte_rigidity_over_C} to give a lower bound on the sparsity $s = N^2 - \abs{I}$ in terms of the rank $r$ and the largest singular value $\sigma_1$ of $A$. As $A$ and $\LinC$ are both binary matrices taking values from $\BK{-1, 1}$, we have
    \begin{align*}
        s = \nnz(A - \LinC) = \frac{1}{4}\sum_{(i,j) \in [N]^2} \bk*{A[i,j] - \LinC[i,j]}^2
        = \frac{1}{4}\sum_{(i,j) \in [N]^2} \bk*{2 - A[i,j] \LinC[i,j]}. \numberthis \label{eq:s < sum(2 - AL) in sec_rigidity_lb}
    \end{align*}
    As $\LinC$ is a rank-$\rinC$ matrix with $\LinC = \UinC^{\top} \VinC$, we can write $\LinC = \sum_{k=1}^{\rinC} {\vecUinCrow}^{\top} \vecVinCrow$, where $\vecUinCrow$ and $\vecVinCrow$ are the $k$-th row vector of $\UinC$ and $\VinC$, respectively. Moreover, for any $k \le \rinC$, 
    \begin{align*}
        &\abs[BBig]{\sum_{(i,j) \in [N^2]} A[i,j] \bk*{{\vecUinCrow}^{\top} \vecVinCrow}[i,j]}
        =\abs[BBig]{\sum_{(i,j) \in [N^2]} A[i,j] {\vecUinCrow}[i] \vecVinCrow[j]}\\
        {}={}& \abs*{\vecUinCrow A \vecVinCrow^{\top}}
        {}\le{} \norm[Big]{\vecUinCrow}_2 \norm[Big]{A \vecVinCrow^{\top}}_2
        {}\le{} \norm[Big]{\vecUinCrow}_2 \cdot \bk*{\sigma_1 \norm[Big]{\vecVinCrow^{\top}}_2}
        {}\le{} \sigma_1 C^{2r} N,
    \end{align*}
    where the first inequality is the Cauchy-Schwartz inequality, the second inequality uses the definition of the largest singular value $\sigma_1$ and the third inequality is because each entry of $\vecUinCrow$ and $\vecVinCrow$ is bounded by $C^r$, hence $\norm[Big]{\vecUinCrow}$ and $\norm[Big]{\vecVinCrow}$ are both bounded by $C^r \sqrt{N}$. Summing over all $k$'s, we get 
    \begin{align*}
        \abs[BBig]{\sum_{(i,j) \in [N^2]} A[i,j] \LinC[i,j]}
        \le \sum_{k=1}^{\rinC} \abs[BBig]{\sum_{(i,j) \in [N^2]} A[i,j] \bk*{{\vecUinCrow}^{\top} \vecVinCrow}[i,j]}
        \le \sigma_1 C^{2r} \rinC N. \numberthis \label{eq: sum(AL) < C^r in sec_rigidity_lb}
    \end{align*}
    Plugging \eqref{eq: sum(AL) < C^r in sec_rigidity_lb} to \eqref{eq:s < sum(2 - AL) in sec_rigidity_lb}, we get 
    \begin{align*}
        s \ge \frac{N^2}{2} - \frac{1}{4} \abs[BBig]{\sum_{(i,j) \in [N^2]} A[i,j] \LinC[i,j]}
        \ge N^2 \bk*{\frac{1}{2} - \frac{\sigma_1 C^{2r} \rinC}{4N} }. \numberthis \label{eq:s < N^2(1/2-tiny) in sec_rigidity_lb}
    \end{align*}
    As $\rinC \defeq \bk{p^3 + 1}^r$ is also exponential in $r$, there is a constant $c$ (depending on $p$) such that $C^{2r} \rinC \le c^r$ for any $r \in \N$. Hence, we can conclude the desired result from \eqref{eq:s < N^2(1/2-tiny) in sec_rigidity_lb} that
    \begin{align*}
        &\boolRigidity{A}{r} \ge N^2 \bk*{\frac{1}{2} - \frac{c^r \sigma_1}{N}}. \qedhere
    \end{align*}
\end{proofof}

\subsection{Rigidity Lower Bound for Kronecker Matrices}
In this subsection, we apply \cref{thm:rigidity_lb_from_singular_value} to Kronecker matrices to get nearly tight rigidity lower bounds for small ranks.

\KroneckerLb*

    According to \cref{thm:rigidity_lb_from_singular_value}, in order to prove Theorem~\ref{thm:kronecker_lb}, it suffices to bound the largest singular value of $\kro{A}$. The following lemma will help us to do this.
    \begin{lemma}
        \label{lm:singular_value_of_A}
        For any matrix $A \in \BK{-1,1}^{q \times q}$ with $\rank(A) > 1$, the largest singular value $\sigma_1 \defeq \sigma_1(A)$ of $A$ is strictly smaller than $q$.
    \end{lemma}
    \begin{proof}
        By the definition of singular value, $\sigma_1^2$ is one of the eigenvalues of the matrix $B = A^{\top} A$, hence there is an eigenvector $\vec{v} \in \mathbb{C}^{q}$ such that 
        $ B\vec{v} = A^{\top} A \vec{v} = \sigma_1^2 \vec{v}.$ 
        Suppose $\vec{v} = \bk{v_1, \ldots, v_q}$ and assume $j \in [q]$ is the index maximizing $\abs{v_j}$. Then,
        \begin{align*}
            \sigma_1^2 \abs{v_j} = \abs[BBig]{\bk*{B \vec{v}}[j]}
            = \abs*{\sum_{i=1}^{q} B[i,j] v_i}
            \le \sum_{i=1}^q \abs[BBig]{B[i,j]} \abs[Big]{v_i}
            \le \sum_{i=1}^q \abs[BBig]{B[i,j]} \abs[Big]{v_j}.
        \end{align*}
        Hence, $\sigma_1^2 \le \sum_{i=1}^q \abs{B[i,j]}$. Moreover, as $\abs*{B[i,j]} = \abs*{\sum_{\l = 1}^q A[\l, i] A[\l, j]} \le q$ for any $i \in [q]$, we get $\sigma_1^2 \le q^2$, i.e., $\sigma_1 \le q$.

        We further check the equality $\sigma_1 = q$ cannot hold. Otherwise, we have $\abs*{\sum_{\l = 1}^q A[\l, i] A[\l, j]} = q$, which means $A[\l, i] = A[\l, j]$ for all $i, \l \in [q]$, i.e., all the columns of $A$ are same as the $j$-th column. This means that $A$ has rank 1, a contradiction. Hence, we have $\sigma_1 < q$, as desired.
    \end{proof}

\begin{proofof}{\cref{thm:kronecker_lb}}
    Using \cref{lm:singular_value_of_A}, the largest singular value $\sigma_1(\kro{A})$ of $\kro{A}$ can be represented as $\sigma_1^n$, where $\sigma_1 < q$ is the largest singular value of $A$. Let $c_1 > 0$ be a parameter to be determined. By \cref{thm:rigidity_lb_from_singular_value}, the rigidity of $\kro{A}$ can be bounded as 
    \begin{align*}
        \boolRigidity{\kro{A}}{c_1 n} 
        \ge q^n \bk*{\frac{1}{2} - \frac{c^{c_1 n} \cdot \sigma_1^n}{q^n}}
        = q^n \bk*{\frac{1}{2} - \bk*{\frac{c^{c_1}\sigma_1}{q}}^n}
    \end{align*}
    for some constant $c>0$ which is defined in \cref{thm:rigidity_lb_from_singular_value}. As $\sigma_1 < q$, there is a sufficiently small constant $c_1 > 0$ such that $c_2 \defeq {c^{c_1}\sigma_1}/{q} < 1$, which implies the desired result.
\end{proofof}

\subsection{Rigidity Lower Bound for the Distance Matrix}
For any $n \in \N$, we define the $2^n \times 2^n$ matrix $M_n$ called the \defn{distance matrix} as $M_n \defeq \maj{A}$, where $A = \bk*{\begin{matrix}
    1 & -1\\
    -1 & 1
\end{matrix}}$. In this subsection, we prove the following rigidity lower bound for $M_n$:
\HammingLb*

    Similar to the previous subsection, we need to bound the largest singular value of $M_k$. Here we explicitly compute the singular values as follows:
    \begin{claim}
    \label{clm:list_eigenvalues}
        All the eigenvalues of $M_n$ are listed below:
        \begin{align*}
            \lambda_y = \sum_{z \in \BK{0,1}^n} \bk*{2\indicator{|z| \le n/2}-1} (-1)^{\angbk{y,z}},\quad \forall y \in \BK{0,1}^n.
        \end{align*}
    \end{claim}
    \begin{proof}
        By the definition of the distance matrix $M_n = \maj{A}$, we can represent each entry of $M_n$ as 
        \begin{align*}
            M_n[x,y] = \Maj\bk{A[x_1,y_1], \ldots, A[x_n, y_n]} = 2\indicator{|{x-y}| \le n/2}-1
        \end{align*}
        for each $x, y \in \BK{0,1}^n$. Below, we check that $\lambda_y$ is an eigenvalue of $M_n$, with associated eigenvector $v_y \in \BK{-1,1}^{2^n}$ defined as $v_y[x] = (-1)^{\angbk{y,x}}$ for any $x \in \BK{0, 1}^n$: 
        \begin{align*}
            M_n v_y [x] 
            {}={}& \sum_{z \in \BK{0,1}^n} M_n[x,z] v_y[z]
            = \sum_{z \in \BK{0,1}^n} \bk*{2\indicator{|{x-z}| \le n/2}-1} (-1)^{\angbk{y,z}}\\
            {}={}& (-1)^{\angbk{y,x}}\sum_{z \in \BK{0,1}^n} \bk*{2\indicator{|{x\oplus z}| \le n/2}-1} (-1)^{\angbk{y,x\oplus z}}\\
            {}={}& v_y[x]\sum_{z \in \BK{0,1}^n} \bk*{2\indicator{|{z}| \le n/2}-1} (-1)^{\angbk{y,z}}
            {}={} \lambda_y v_y[x].
        \end{align*}
        Here, we use $x \oplus z$ to denote the bit-wise XOR of $x$ and $z$: $(x_1 \oplus z_1, \ldots, x_n \oplus z_n)$. As the vectors $v_y$ form an orthogonal basis, we have enumerated all the eigenvalues of $M_n$ as  $\BK{\lambda_y : y \in \BK{0,1}^n}$.
    \end{proof}

Now that we have listed all the singular values of $M_n$, we can bound the largest one.

    \begin{claim}
    \label{clm:singular_bound_hamming}
        The largest singular value $\sigma_1(M_n)$ of $M_n$ is bounded above by $O\bk{2^n/\sqrt{n}}$.
    \end{claim}
    \begin{proof}
        As $M_n$ is a symmetric matrix, we only need to show that for any eigenvalue $\lambda_y$ defined in \cref{clm:list_eigenvalues}, we have $\abs{\lambda_y} = O\bk{2^n/\sqrt{n}}$.

        If $y$ is the zero vector, i.e., $y = (0, \ldots, 0)$, then 
            $\lambda_y = \sum_{z \in \BK{0,1}^n} \bk*{2\indicator{|{z}| \le n/2}-1}.$
        For each $z \in \BK{0,1}^n$ with $|{z}| < n/2 $, we pair it with $\overline{z} \defeq (1-z_1, \ldots, 1-z_n)$ in the sum above, so that the paired summands will cancel out since  
        $\bk*{2\indicator{|{z}| \le n/2}-1} + \bk*{2\indicator{|{\overline{z}}| \le n/2}-1} = 0.$
        Hence, 
        \begin{align*}
            |\lambda_y| 
            = \abs*{\sum_{\substack{z \in \BK{0,1}^n, |z| = n/2}} \bk*{2\indicator{|{z}| \le n/2}-1}}
            \le {\sum_{\substack{z \in \BK{0,1}^n\\ |z| = n/2} } 1} 
            \le \binom{n}{n/2} = O\bk*{2^n/\sqrt{n}}.
        \end{align*}

        Otherwise, $y$ is non-zero, say, $y_1 \neq 0$. Then, for each $z\in \BK{0,1}^n$ with $|z| \notin \Bk*{\frac{n}{2}-1, \frac{n}{2}+1}$, we pair it with $z^{\oplus 1} \defeq (1 - z_1, z_2, \ldots, z_n)$. Thus, $\indicator{|z| \le n/2} = \indicator{|z^{\oplus 1}| \le n/2}$ and $(-1)^{\angbk{y,z}} + (-1)^{\angbk{y, z^{\oplus 1}}} = 0$, hence their corresponding terms in the sum defining $\lambda_y$ again cancel out. Thus,
        \begin{align*}
            |\lambda_y| 
            =& \abs*{\sum_{\substack{z \in \BK{0,1}^n, |z| \in \Bk*{\frac{n}{2} - 1, \frac{n}{2} + 1}}} \bk*{2\indicator{|{z}| \le n/2}-1}}
            \le {\sum_{\substack{z \in \BK{0,1}^n\\ |z| \in \Bk*{\frac{n}{2} - 1, \frac{n}{2} + 1} } }1} 
            = O\bk*{2^n/\sqrt{n}}. \qedhere
        \end{align*}
    \end{proof} 
\begin{proofof}{\cref{thm:hamming_lb}}
    Plugging the singular value bound in \cref{clm:singular_bound_hamming} into \cref{thm:rigidity_lb_from_singular_value}, we show that there is a constant $c > 1$, such that for any rank $r$,
    \begin{align*}
        &\boolRigidity{M_n}{r}
        \ge 4^n \bk*{\frac{1}{2} - \frac{c^r \cdot 2^n /\sqrt{n}}{2^n}}
        = 4^n \bk*{\frac{1}{2} - \frac{c^r }{\sqrt{n}}}. 
    \end{align*}
    For any $\eps > 0$, we can take a sufficiently small constant $\beta > 0$ such that $\beta \log c < \eps$, then 
    \begin{align*}
        &\boolRigidity{M_n}{\beta \log n}
        \ge 4^n \bk*{\frac{1}{2} - \frac{c^{\beta \log n} }{\sqrt{n}}}
        > 4^n \bk*{\frac{1}{2} - \frac{1 }{n^{1/2 - \eps}}}. \qedhere
    \end{align*}
\end{proofof}

\section{Hardness Amplification for Rigidity}
\label{sec:amplification}
We now move on to proving our hardness amplification results.

\subsection{Hardness Amplification for Kronecker Products}

\KroneckerAmp*

In the remainder of this subsection, we prove Theorem~\ref{thm:amplification_kro}.

    As $\boolRigidity{A}{r} \le \delta \cdot q^2$, there is a rank-$r$ matrix $L \in \F_p^{q \times q}$, such that
    \begin{align*}
        \bool(A[i,j]) = \bool(L[i,j]), \quad \forall (i,j) \in I,
    \end{align*}
    where $I \subset [q]^2$ is a set of indices with $|I| \ge (1 - \delta) \cdot q^2$. The goal of this proof is to construct a low-rank matrix $\tilde{L} \in \F_p^{q^n \times q^n}$ that approximates $\kro{A}$.

    We index the rows and columns of $\kro{A}$ by vectors in $[q]^n$. Specifically, for any $x = (x_1, \ldots, x_n)$, and  $y = (y_1, \ldots, y_n) \in [q]^n$, the corresponding entry in $A$ is
    \begin{align*}
        \kro{A} [x,y] = A[x_1, y_1] \cdots A[x_n, y_n].
    \end{align*}

    Guided by the intuition introduced in \cref{subsec:overview_amplification_kro}, the matrix $\tilde{L}$ will be the low-rank approximation of the matrix $\kro{L}$ (which further approximates $\kro{A}$), constructed via a degree-$1$ polynomial $\tilde{\pi}$ that approximates the multiplication function $\pi: \F_p^n \to \F_p$, defined as follows.
    \begin{lemma}
    \label{lm:linear_polynomial_approx}
        Let $a$ be a uniformly random vector from $\F_p^n$. Define $\tilde{\pi}_a : \F_p^n \to \F_p$ as
        \begin{align*}
            \tilde{\pi}_a(z_1, z_2, \ldots, z_n) \defeq 1 + \sum_{i=1}^n a_i (z_i - 1).
        \end{align*}
        Then, for $z_1, \cdots, z_n \in \F_p$ that are not all $1$'s, over all random seeds $a$,
        \begin{align*}
            \Pr_{a \in \F_p^n}\Bk*{\bool\bk*{\tilde{\pi}_a (z_1, \cdots, z_n)} = 1} = \frac{1}{p}. 
        \end{align*}
        Moreover, for $z_1 = \cdots = z_n = 1$, 
        \begin{align*}
            \Pr_{a \in \F_p^n}\Bk*{\bool\bk*{\tilde{\pi}_a (z_1, \cdots, z_n)} = 1} = 1. 
        \end{align*}
    \end{lemma}
    \begin{proof}
        When $z_1, \ldots, z_n$ are not all $1$, $\tilde{\pi}_a(z_1, z_2, \ldots, z_n)$ is uniformly distributed over $\F_p$ over all random seeds $a$, hence $\Pr_{a \in \F_p^n}\Bk*{\bool\bk*{\tilde{\pi}_a (z_1, \cdots, z_n)} = 1} = {1}/{p}$. When $z_1 = \cdots = z_n = 1$, $\tilde{\pi}_a(z_1, z_2, \ldots, z_n)$ takes a fixed value $1$ hence $\Pr_{a \in \F_p^n}\Bk*{\bool\bk*{\tilde{\pi}_a (z_1, \cdots, z_n)} = 1} = 1. $
    \end{proof}

    Now, the low-rank matrix $\tilde{L}$ is defined as 
    \begin{align*}
        \tilde{L}[x,y] \defeq \tilde{\pi}_a \bk{L[x_1, y_1], \ldots, L[x_n, y_n]}
    \end{align*}
    for some seed $a \in \F_p^n$ to be determined. From the definition of $\tilde{\pi}_a$ in Lemma~\ref{lm:linear_polynomial_approx} above, we see that $\rank\bk{\tilde{L}} \le nr + 1 \le 2nr$, as $\tilde{L}$ can be represented as a linear combination of $n$ matrices of rank $r$ and an all-one matrix. Below, we prove that $\tilde{L}$ approximates $\kro{A}$ well in expectation over the choice of the seed $a$. Specifically, define
    \begin{align*}
        s \defeq \abs*{\BK*{(x,y) : \bool\bk{\tilde{L}[x,y]} \neq \bool\bk{\kro{A}[x,y]}}}, 
    \end{align*}
    and we will show that 
        $\E_{a \in \F_p^n}\Bk{s} \le q^{2n}\bk*{\frac{1}{2} - \frac{1}{2} \cdot \bk*{\frac{1}{2} - \alpha - \delta}^n},$
    i.e., 
    \begin{align*}
    \label{ineq:low-rank-approximates-well-to-prove}
        \Pr_{\substack{(x,y) \in [q]^{n} \times [q]^n \\ a \in \F_p^n}}\Bk*{\bool\bk{\tilde{L}[x,y]} \neq \bool\bk{\kro{A}[x,y]}} \le \frac{1}{2} - \frac{1}{2} \cdot \bk*{\frac{1}{2} - \alpha - \delta}^n, \numberthis
    \end{align*}
    where we view $x,y$ as uniformly chosen random variables from $[q]^{n}$. 
    Assuming that \eqref{ineq:low-rank-approximates-well-to-prove} holds, then there must be a specific choice of $a \in \F_p^n$ which achieves at most the expected number of errors, i.e., such that $s \le q^{2n}\bk*{\frac{1}{2} - \frac{1}{2} \cdot \bk*{\frac{1}{2} - \alpha - \delta}^n}$, and so the matrix $\tilde{L}$ with that seed $a$ is a good low-rank approximation of $\kro{A}$, which concludes the desired rigidity upper bound for $\kro{A}$. Hence, in order to conclude our proof of Theorem~\ref{thm:amplification_kro}, it suffices to prove \eqref{ineq:low-rank-approximates-well-to-prove}.

    We conclude \eqref{ineq:low-rank-approximates-well-to-prove} from the following lemma, where the random variables $A_i$'s and $L_i$'s are the Booleanization of matrix entires $A[x_i,y_i]$ and $L[x_i, y_i]$, respectively, with random $(x,y) \in [q]^n \times [q]^n$.
    \begin{lemma}
        Let $(A_1, L_1), \ldots, (A_n, L_n) \in \BK{-1, 1}^2$ be  independent pairs of random variables, which are identically distributed with $\Pr\Bk{A_i \neq L_i} \le \delta$ and $\abs{\Pr\Bk{A_i = 1} - \Pr\Bk{A_i = -1}} \le \alpha$ for each $i$. Let $a \in \F_p^n$ be drawn uniformly at random. Then,
        \begin{align*}
        \label{ineq:approx_in_RVs}
            \Pr\Bk{\tilde{\pi}_a(L_1, \ldots, L_n) \neq \pi(A_1, \ldots, A_n)} \le \frac{1}{2} -\frac{1}{2} \cdot \bk*{\frac{1}{2} - \alpha - \delta}^n. \numberthis
        \end{align*}
    \end{lemma}
    \begin{proof}
        Let $K$ be the number of $-1$'s in $A_1, \ldots, A_n$ and $L$ be the number of $-1$'s in $L_1, \ldots, L_n$. Let $p_j \defeq \Pr\Bk{A_i = j}$ and $\delta_j = \Pr\Bk{L_i \neq A_i \mid A_i = j}$ for $j \in \BK{-1,1}$. Then, $\abs{p_1 - p_{-1}} \le \alpha$ and $p_1 \delta_1 + p_{-1} \delta_{-1} \le \delta$.

        To compute the probability in the left hand side of \eqref{ineq:approx_in_RVs}, we enumerate the possible values of $K$ and $T$:
        \begin{itemize}
            \item For any $k \le n$, $\Pr\Bk{K = k} = \binom{n}{k} p_{-1}^k p_{1}^{n-k}$. Given that $K = k$, it follows that $\pi(A_1, \ldots, A_n)=(-1)^k$.
            \item Conditioned on $K = k$, we have $\Pr\Bk{T = 0 \mid K = k} = \delta_{-1}^k \bk{1 - \delta_1}^{n-k}$.
            \item Conditioned on $T = 0$ and $K = k$, according to \cref{lm:linear_polynomial_approx}, we have $\tilde{\pi}_a(L_1, \ldots, L_n) = 1$ with probability 1, which means $\tilde{\pi}_a(L_1, \ldots, L_n) = \pi(A_1, \ldots, A_n)$ with probability $\frac{1 + (-1)^k}{2}$.
            \item Conditioned on $T \neq 0$ and $K = k$, according to \cref{lm:linear_polynomial_approx}, we have $\tilde{\pi}_a(L_1, \ldots, L_n) = 1$ with probability $1/p$, which means $\tilde{\pi}_a(L_1, \ldots, L_n) = \pi(A_1, \ldots, A_n)$ with probability $\frac{1}{p} \frac{1 + (-1)^k}{2} + \frac{p-1}{p} \frac{1 - (-1)^k}{2} = \frac{1}{2} + \frac{(2-p)(-1)^k}{2p}$.
        \end{itemize}
        Combining these probabilities together, for any $k \le n$, we have 
        \begin{align*}
            &\Pr\Bk{\tilde{\pi}_a(L_1, \ldots, L_n) = \pi(A_1, \ldots, A_n) \mid K = k}\\
            {}={}& \Pr\Bk{T = 0 \mid K = k} \cdot \frac{1 + (-1)^k}{2} + \Pr\Bk{T \neq 0 \mid K = k} \cdot \bk*{\frac{1}{2} + \frac{(2-p)(-1)^k}{2p}}\\
            {}={}& \delta_{-1}^k \bk{1 - \delta_1}^{n-k} \cdot \frac{1 + (-1)^k}{2} + \bk*{1 - \delta_{-1}^k \bk{1 - \delta_1}^{n-k}} \cdot \bk*{\frac{1}{2} + \frac{(2-p)(-1)^k}{2p}}\\
            {}={}& \frac{1}{2} + \frac{(2-p)(-1)^k}{2p} + \delta_{-1}^k \bk{1 - \delta_1}^{n-k} (-1)^k \cdot \frac{p-1}{p},
        \end{align*}
        hence,
        \begin{align*}
            &\Pr\Bk{\tilde{\pi}_a(L_1, \ldots, L_n) \neq \pi(A_1, \ldots, A_n)}\\
            {}={}& \sum_{k=1}^n \Pr\Bk{K = k} \cdot \Pr\Bk{\tilde{\pi}_a(L_1, \ldots, L_n) \neq \pi(A_1, \ldots, A_n) \mid K = k}\\
            {}={}& \sum_{k=1}^n \binom{n}{k}p_{-1}^k p_1^{n-k} \bk*{\frac{1}{2} - \frac{(2-p)(-1)^k}{2p} - \delta_{-1}^k \bk{1 - \delta_1}^{n-k} (-1)^k \cdot \frac{p-1}{p}}\\
            {}={}& \frac{1}{2} - \frac{2-p}{2p}\bk*{-p_{-1} + p_1}^n - \frac{p-1}{p}\bk*{p_1(1 - \delta_1) - p_{-1}\delta_{-1}}^n\\
            {}\le{}& \frac{1}{2} + \frac{p-2}{2p}\alpha^n - \frac{p-1}{p}\bk*{\frac{1}{2} - \alpha - \delta}^n
            {}\le{} \frac{1}{2}  - \frac{1}{2} \cdot \bk*{\frac{1}{2} - \alpha - \delta}^n,
        \end{align*}
        where we used that $2\alpha + \delta < 1/2$ which implies $\alpha < 1/2 - \alpha - \delta$.
    \end{proof}

This concludes the proof of Theorem~\ref{thm:amplification_kro}.

\subsection{Hardness Amplification for Majority Product}

\MajorityAmp*

\begin{proof}
    As $\maj[k]{A}$ has probabilistic boolean rank $r$, there is a distribution $\mathcal{L}$ on matrices in $\F_p^{q^k \times q^k}$ with rank $r$, such that for any $(i,j) \in [q]^k \times [q]^k$ and a random matrix $L$ sampled from the distribution $\mathcal{L}$, 
    \[
    \label{ineq:probabilistic_rank}
        \Pr_{L \sim \mathcal{L}} \Bk*{\bool\bk{L[i,j]} \neq \bool\bk{\maj[k]{A}[i,j]}} \le \delta. \numberthis
    \]
    Below, we define a rank-$r$ matrix $\tilde{L} \in \F_p^{q^n \times q^n}$ that approximates $\maj{A}$ well.

    Recall the indices of rows and columns of $\maj{A}$ can be viewed as vectors in $[q]^n$. For any $x = (x_1, \ldots, x_n) \in [q]^n$, we define $x_\pre = (x_1, \ldots, x_k) \in [q]^k$ as the prefix of $x$ of length $k$. We define $\tilde{L}$ as 
    \begin{align*}
        \tilde{L}[x,y] = L[x_\pre, y_\pre], \quad \forall (x,y) \in [q]^n \times [q]^n,
    \end{align*}
    where $L$ is a matrix to be determined from the support of the distribution $\mathcal{L}$. By construction, $\tilde{L}$ has the same rank $r$ as $L$. To conclude a Boolean rigidity upper bound of $\maj{A}$, we need to show that if $L$ is drawn from $\mathcal{L}$, then  $\bool\bk{\tilde{L}[x,y]}$ disagrees with $\bool\bk{\maj{A}[x,y]}$ in at most $q^{2n} \bk*{\frac{1}{2} - \Omega\bk*{(1 - 2\delta) \cdot \frac{\sqrt{k}}{\sqrt{n}}}}$ entries in expectation, i.e.,
    \begin{align*}
    \label{ineq:low_rank_for_maj_to_prove}
        \Pr_{\substack{(x,y) \in [q]^n \times [q]^n\\L \sim \mathcal{L}}}\Bk*{\bool\bk{\tilde{L}[x,y]} \neq \bool\bk{\maj{A}[x,y]}} \le \frac{1}{2} -  \Omega\bk*{(1 - 2\delta) \cdot \frac{\sqrt{k}}{\sqrt{n}}}.\numberthis
    \end{align*}
    We can then fix $L$ be the best matrix in the support of $\mathcal{L}$ to achieve the desired error.
    
    To prove \eqref{ineq:low_rank_for_maj_to_prove}, recall that 
        $\maj{A}[x,y] = \Maj \bk{A[x_1, y_1], \ldots, A[x_n, y_n]}.$
    Let  $A_1, \ldots, A_n \in \BK{-1,1}$ be random variables obtained by taking random $(x,y) \in [q]^n\times[q]^n$ and then setting $A_i = A[x_i, y_i]$ for all $i$. Then, $A_1, \ldots, A_n$ are independent random variables with $\Pr\Bk{A_i = 1} = 1/2$.
    Moreover, by \eqref{ineq:probabilistic_rank}, for any possible choice of $A_1, \ldots, A_k$, over the randomness of $L$, we have 
    \[
        \Pr_{L \sim \mathcal{L}} \Bk*{\bool\bk{L[x_\pre, y_\pre]} \neq \Maj\bk{A_1, \ldots, A_k}} \le \delta.
    \]
    Hence, there is a random variable $\Delta \in \BK{-1,1}$ such that $\bool\bk{L[x_\pre, y_\pre]} = \Maj\bk{A_1, \ldots, A_k} \cdot \Delta$, with $\Pr_{L \sim \mathcal{L}}\Bk{\Delta = -1} \le \delta$ \emph{conditioned on} any possible choice of $A_1, \ldots, A_k$.
    
    We conclude \eqref{ineq:low_rank_for_maj_to_prove} from the following lemma.
    \begin{lemma}
         Let $A_1, \ldots, A_n \in \BK{-1, 1}$ be identically-distributed  independent random variables with $\Pr\Bk{A_i = 1} = 1/2$. Let $\Delta \in \BK{-1, 1}$ be a random variable which is independent of $A_{k+1}, \ldots, A_n$, such that $\Pr\Bk{\Delta = -1 | A_1 = A_1^*, \ldots, A_k = A_k^*} \le \delta$ for any possible choices of $A_1^*, \ldots, A_k^*$. Then, 
        \begin{align*}
        \label{ineq:approx_maj_in_RVs}
            \Pr\Bk*{\bk{\Delta \cdot \Maj\bk{A_1, \ldots, A_k}}  \neq \Maj\bk{A_1, \ldots, A_n}} \le \frac{1}{2} - \Omega\bk*{(1 - 2\delta) \cdot \frac{\sqrt{k}}{\sqrt{n}}}. \numberthis
        \end{align*}
    \end{lemma}
    \begin{proof}
        The first step of the proof is to assume without loss of generality that $\Pr\Bk{\Delta = -1 \mid A_1 = A_1^*, \ldots, A_k = A_k^*} = \delta$, i.e., $\Delta$ is independent to $A_1, \ldots, A_n$ with $\Pr\Bk{\Delta  = -1} = \delta$. This is because, conditioned on any possible $A_1^*, \ldots, A_k^*$, if we let $p^* \ge 1/2$ denote the probability that $\Maj\bk{A_1, \ldots, A_n}$ agrees with $\Maj\bk{A_1, \ldots, A_k}$ and $\delta^* \le \delta$ denote the probability that $\Delta = -1$, then
        \begin{align*}
            &\Pr\Bk*{\bk{\Delta \cdot \Maj\bk{A_1, \ldots, A_k}}  \neq \Maj\bk{A_1, \ldots, A_n} \mid A_1 = A_1^*, \ldots, A_k = A_k^*}\\
            {}={} &(1 - \delta^*) (1 - p^*) + \delta^* p^*
            {}={}  (1 - p^*) + (2p^* - 1) \delta^*.
        \end{align*}
        This probability is maximized when $\delta^*$ is equal to $\delta$. Hence, we can assume without loss of generality that $\Delta$ is independent of $A_1, \ldots, A_n$ and $\Pr\Bk{\Delta = -1} = \delta$. Then,
        \begin{align*}
        \label{ineq:prob_with_independent_error}
            &\Pr\Bk*{\bk{\Delta \cdot \Maj\bk{A_1, \ldots, A_k}}  \neq \Maj\bk{A_1, \ldots, A_n}} = (1-p) + (2p-1) \delta, \numberthis
        \end{align*}
        where $p \defeq \Pr\Bk{\Maj\bk{A_1, \ldots, A_k} = \Maj\bk{A_1, \ldots, A_n}}$. We use the following fact about binomial random variables to bound $p$:
        \begin{fact}
        \label{fact:concentration_of_binomial}
            Let $X_1, \ldots, X_n \in \BK{-1,1}$ be identically-distributed independent random variables with $\Pr\Bk{X_i = 1} = 1/2$. Let $X\defeq X_1 + \cdots + X_n$. Then, for any $a \ge 1$,
            \begin{align*}
                \Pr\Bk{X \ge a} = \frac{1}{2} - \Theta\bk*{\frac{a}{\sqrt{n}}}.
            \end{align*}
        \end{fact}
        To bound $p$, let $\mathcal{E}$ be the event that $A_1+ \cdots + A_k \ge \eps\sqrt{k}$ where $\eps >0$ is a small constant
        so that, by applying \cref{fact:concentration_of_binomial} on $A_1, \ldots, A_k$, we get $\Pr\Bk{\mathcal{E}}\ge \Omega(1)$. Moreover, conditioned on $\mathcal{E}$ happening, we know that $\Maj\bk{A_1, \ldots, A_k} = 1$, and so
        \begin{align*}
            &\Pr\Bk{\Maj\bk{A_1, \ldots, A_n} = 1}
            = \Pr\Bk{A_1 + \cdots + A_n \ge 0}\\
            {}\ge{} &\Pr\Bk*{A_{k+1} + \cdots + A_n \ge -\eps{\sqrt{k}}}
            {}\ge{} \frac{1}{2} + \Omega\bk*{\frac{\sqrt{k}}{\sqrt{n}}},
        \end{align*}
        where in the last inequality we apply \cref{fact:concentration_of_binomial} over $-A_{k+1}, \ldots, -A_{n}$. Hence,
        \begin{align*}
            \Pr\Bk[\big]{\Maj\bk{A_1, \ldots, A_k} = \Maj\bk{A_1, \ldots, A_n} \mid \mathcal{E}} \ge \frac{1}{2} + \Omega\bk*{\frac{\sqrt{k}}{\sqrt{n}}}.
        \end{align*}
        On the other hand, conditioned on $\mathcal{E}$ not happening, we still have 
        \begin{align*}
            \Pr\Bk[\big]{\Maj\bk{A_1, \ldots, A_k} = \Maj\bk{A_1, \ldots, A_n} \mid \overline{\mathcal{E}}} \ge \frac{1}{2}.
        \end{align*}
        Combining these probabilities together, we get
        \begin{align*}
            p = \Pr\Bk[\big]{\Maj\bk{A_1, \ldots, A_k} = \Maj\bk{A_1, \ldots, A_n}} 
            \ge \frac{1}{2} + \Pr\Bk{\mathcal{E}} \cdot \Omega\bk*{\frac{\sqrt{k}}{\sqrt{n}}} 
            \ge \frac{1}{2} + \Omega\bk*{\frac{\sqrt{k}}{\sqrt{n}}}.
        \end{align*}
        Substituting this bound into \eqref{ineq:prob_with_independent_error}, we get 
        \begin{align*}
            &\Pr\Bk*{\bk{\Delta \cdot \Maj\bk{A_1, \ldots, A_k}}  \neq \Maj\bk{A_1, \ldots, A_n}}\\
            {}={}& (1-p) + (2p-1) \delta
            {}={} \frac{1}{2} - \bk*{p - \frac{1}{2}} ( 1-2\delta)\\
            {}\le{}& \frac{1}{2} - \Omega\bk*{(1 - 2\delta) \cdot \frac{\sqrt{k}}{\sqrt{n}}}. \qedhere
        \end{align*}
    \end{proof}
\end{proof}

\section{Discussion}
\label{sec:discussion}
In this section, we will discuss the implications of our results. 

\subsection{Slightly Improved Rigidity Lower Bounds Would Yield Razborov Rigidity}
In \cref{sec:lower_bound}, we proved nearly tight rigidity lower bound for Kronecker powers (\cref{thm:kronecker_lb}) and the distance matrices (\cref{thm:hamming_lb}) in the small-rank regime. It is natural to ask whether we can improve the rigidity lower bound such that it works for better rank or sparsity parameters. However, the following theorems show that such an improvement may be very challenging: even improving our rigidity lower bounds slightly would imply a Razborov rigidity lower bound, a famous hard problem in the fields of matrix rigidity and communication complexity.

\StrongerKro*

\StrongerMaj*

\begin{proofof}{\cref{thm:stronger_lb_kro_imply_Razborov}}
    Assume to the contrary that $\BK{A_n}_{n \in \mathbb{N}}$ is not Razborov rigid. Thus, there is a constant $c > 0$, such that for any $n$, we have $\boolRigidity{\kro{A}}{2^{\log^c n}} \le \frac{1}{3} \cdot q^{2n}$. Let $k < n$ be a parameter to be determined. We can use \cref{thm:amplification_kro} on the matrix $\kro[k]{A}$ and its $(n/k)$-th Kronecker power $\kro{A}$: as $\boolRigidity{\kro[k]{A}}{2^{\log^c k}} \le \frac{1}{3} \cdot q^{2k}$, we get
    \begin{align*}
        \boolRigidity{\kro{A}}{2^{1 + \log^c k} \cdot n} \le q^{2n} \bk*{\frac{1}{2} - \frac{1}{2} \cdot \bk*{\frac{1}{6} - \alpha}^{n/k}},
    \end{align*}
    where $\alpha$ is the difference between the fraction of $1$'s and $-1$'s in $A^{\otimes k}$. In particular, since $A$ is not the all-$1$'s or all-$(-1)$'s matrix (since its rank is not $1$), it follows that for sufficiently large $k$, we have $\alpha < 1/12$.\footnote{One can show $\alpha < 1/12$ by directly counting, although it also follows from our rigidity lower bound above. We need to show that there are roughly the same number of $1$ and $-1$ entries in $\kro[k]{A}$. To argue this, one can use \cref{thm:kronecker_lb} on $\kro[k]{A}$ to get $\boolRigidity{\kro[k]{A}}{1} \ge \boolRigidity{\kro[k]{A}}{c_1 k} \ge q^{2k} \bk*{\frac{1}{2} - c_2^k}$, which means both the all-$1$'s and all-$(-1)$'s matrices are bad approximations of $\kro[k]{A}$, so there are roughly the same number of $1$ and $-1$ entries in $\kro[k]{A}$.} Picking $k \defeq 2^{\bk{\eps\log n / 2}^{1/c}}$, so that $2^{\log^c k} = n^{\eps/2}$, we get
    \begin{align*}
        \boolRigidity{\kro{A}}{n^{1 + \eps}} 
        \le \boolRigidity{\kro{A}}{2^{1 + \log^c k} \cdot n} 
        \le q^{2n} \bk*{\frac{1}{2} - \frac{1}{2} \cdot {\frac{1}{12^{n/ 2^{\bk{\eps\log n / 2}^{1/c}}}}}}
        < q^{2n} \bk*{\frac{1}{2} - \frac{1}{2^{n/2^{(\log n)^{o(1)}}}}},
    \end{align*}
    contradicting \eqref{ineq:kro_rigidity_lb_for_larger_rank}.
\end{proofof}

\begin{proofof}{\cref{thm:stronger_lb_maj_imply_Razborov}}
    Simialr to the proof of \cref{thm:stronger_lb_kro_imply_Razborov}, we assume to the contrary that $\BK{M_n}_{n \in \mathbb{N}}$ is not Razborov rigid with $\boolRigidity{M_n}{2^{\log^c n}} \le \frac{1}{3} \cdot 4^{n}$ for some constant $c > 0$, and apply \cref{thm:amplification_maj} on the matrices $M_k$ and $M_n$ to get
    \begin{align*}
        \boolRigidity{M_n}{2^{\log^c k}} \le 4^{n} \bk*{\frac{1}{2} - \Omega\bk*{\frac{\sqrt{k}}{\sqrt{n}}}}.
    \end{align*}
    Taking $k = 2^{\bk{\log\log n + \log \beta}^{1/c}}$, we have $2^{\log^c k} = \beta \log n$, and 
    \begin{align*}
        \boolRigidity{M_n}{\beta\log n}
        \le 4^{n}\bk*{\frac{1}{2} - \Omega\bk*{\frac{2^{\frac{1}{2}\bk{\log\log n + \log \beta}^{1/c}}}{n^{1/2}}}}
        < 4^n \bk*{\frac{1}{2} - \frac{2^{\bk*{\log \log n}^{o(1)}}}{n^{1/2}}},
    \end{align*}
    contradicting \eqref{ineq:maj_rigidity_lb_for_larger_rank}.
\end{proofof}

\subsection{Obstructions to Constructing Depth-2 Circuits via Rigidity} \label{sec:obstructiondepth2}

One important application of matrix rigidity in the small-rank regime is that rigidity upper bounds for matrix $A$ can be used to construct small depth-$d$ linear circuits to compute the Kronecker power $\kro{A}$ of $A$. Specifically, \cite{alman2021kronecker} proved that
\begin{lemma}[\cite{alman2021kronecker}, Lemma 1.4]
    \label{lm:depth-2_circuit}
    Let $A \in \BK{-1,1}^{q\times q}$ be a matrix over some field $\F_p$. For any rank parameter $r \le q$, we define
    \begin{align*}
        c \defeq \log_q \bk{(r + 1) \cdot (r + \R_{A}(r) / q)}.
    \end{align*}
    Then, for any positive integer $n$, the Kronecker power $\kro{A}$ has a depth-$d$ synchronous circuit of size $O(d \cdot q^{(1 + c/d)n})$.
\end{lemma}

By applying this lemma, \cite{alman2021kronecker} obtained a depth-2 synchronous circuit of size $O(2^{1.47582n})$ for Walsh-Hadamard matrix $H_n$ from a rank-1 rigidity upper bound $\R_{H_4}(1) \le 96$, improving on the folklore bound of $O(2^{1.5n})$.

Notably, although most prior work on matrix rigidity had given \emph{asymptotic} upper and lower bounds, it is unclear how to use asymptotic rigidity upper bounds for $H_n$, such as those from~\cite{alman2017probabilistic}, to construct improved circuits in this way (see \cite[Page 2]{alman2021kronecker}). Alman instead used the constant-sized bound $\R_{H_4}(1)$. It was later shown that rank-1 rigidity upper bounds cannot give any further improvements~\cite{alman2023smaller}.

We are able to use our new rigidity lower bound, \cref{thm:kronecker_lb}, to prove that only constant-sized rigidity upper bounds may be effective in conjunction with Lemma~\ref{lm:depth-2_circuit}. We prove that, starting from the matrix $A = H_k$ with a large $k$ and picking any rank $1 < r \le O(k)$, the depth-2 circuit for $H_n$ constructed by \cref{lm:depth-2_circuit} will have size at least $\Omega(2^{3n/2})$.
\begin{lemma}
    Let $c_1, c_2$ be the constants in \cref{thm:kronecker_lb}, and $k$ is a large parameter. For any rank $1 < r \le c_1 k$, we have 
    \begin{align*}
        \log_{2^k} \bk{(r + 1) \cdot (r + \R_{H_k}(r) / 2^k)} \ge 1.
    \end{align*}
\end{lemma}
\begin{proof}
    By \cref{thm:kronecker_lb}, we have $\R_{H_k}(r) \ge \R_{H_k}(c_1 n) \ge 4^k \bk{1/2 - c_2^k} \ge 4^k /3$ for large $k$. Hence, 
    \begin{align*}
        &(r + 1) \cdot (r + \R_{H_k}(r) / 2^k) 
        > 3 \R_{H_k}(r) / 2^k \ge 2^k
    \end{align*}
    as $r \ge 2$, which proves the desried inequality.
\end{proof}

{
\bibliographystyle{alpha}
\bibliography{reference.bib}
}

\end{document}